\documentclass[11pt]{article}

\usepackage[utf8]{inputenc}
\usepackage{graphicx}
\graphicspath{{figs/}}
\IfFileExists{arxiv-figure-deps.tex}{
}{}
\usepackage{indentfirst,csquotes}
\usepackage{amssymb,amsthm,amsmath}
\usepackage{xcolor,paralist,hyperref,titlesec,fancyhdr,etoolbox,float}
\usepackage[font=small,labelfont=bf,skip=6pt]{caption}
\usepackage[margin=1in]{geometry}

\hypersetup{
    colorlinks=true,
    linkcolor=black,
    filecolor=black,
    urlcolor=black
}

\newtheorem{theorem}{Theorem}
\newtheorem{definition}[theorem]{Definition}

\newtheorem{proposition}[theorem]{Proposition}

\newtheorem{remark}[theorem]{Remark}

\titleformat{\section}
  {\normalfont\Large\bfseries}
  {\thesection}{1em}{}

\titleformat{\subsection}
  {\normalfont\large\bfseries}
  {\thesubsection}{1em}{}

\titlespacing{\section}{0pt}{2ex plus .5ex minus .2ex}{1ex plus .2ex}
\titlespacing{\subsection}{0pt}{1.5ex plus .5ex minus .2ex}{1ex plus .2ex}

\DeclareCaptionFont{nestfigcap}{\small}
\newcommand{\nestfigcaption}[2]{%
  \caption{#1}%
  \label{fig:#2}%
}

\newcommand{\nestincludefig}[2][width=\linewidth]{%
  \includegraphics[#1]{#2}%
}

\newcommand{\nestfig}[2][]{%
  \begin{figure}[htbp]
    \centering
    \nestincludefig[width=0.92\linewidth]{#2}%
    \if\relax\detokenize{#1}\relax\else\nestfigcaption{#1}{#2}\fi
  \end{figure}%
}
\newcommand{\nestfigwithlabel}[3][]{%
  \begin{figure}[htbp]
    \centering
    \nestincludefig[width=0.92\linewidth]{#2}%
    \if\relax\detokenize{#1}\relax\else\nestfigcaption{#1}{#3}\fi
  \end{figure}%
}
\newcommand{\nestfigsmall}[2][]{%
  \begin{figure}[htbp]
    \centering
    \nestincludefig[width=0.46\linewidth]{#2}%
    \if\relax\detokenize{#1}\relax\else\nestfigcaption{#1}{#2}\fi
  \end{figure}%
}
\newcommand{\nestfigsmallw}[3][]{%
  \begin{figure}[htbp]
    \centering
    \nestincludefig[width=#3\linewidth]{#2}%
    \if\relax\detokenize{#1}\relax\else\nestfigcaption{#1}{#2}\fi
  \end{figure}%
}
\newcommand{\nestfigload}[5][]{%
  \begin{figure}[htbp]
    \centering
    \begin{minipage}[t]{0.31\linewidth}\centering\nestincludefig[width=\linewidth]{#2}\end{minipage}\hfill
    \begin{minipage}[t]{0.31\linewidth}\centering\nestincludefig[width=\linewidth]{#3}\end{minipage}\hfill
    \begin{minipage}[t]{0.31\linewidth}\centering\nestincludefig[width=\linewidth]{#4}\end{minipage}
    \if\relax\detokenize{#1}\relax\else\nestfigcaption{#1}{#5}\fi
  \end{figure}%
}
\newcommand{\nestfigloadH}[5][]{%
  \begin{figure}[H]
    \centering
    \setlength{\intextsep}{6pt plus 2pt minus 2pt}%
    \begin{minipage}[t]{0.31\linewidth}\centering\nestincludefig[width=\linewidth,height=0.13\textheight,keepaspectratio]{#2}\end{minipage}\hfill
    \begin{minipage}[t]{0.31\linewidth}\centering\nestincludefig[width=\linewidth,height=0.13\textheight,keepaspectratio]{#3}\end{minipage}\hfill
    \begin{minipage}[t]{0.31\linewidth}\centering\nestincludefig[width=\linewidth,height=0.13\textheight,keepaspectratio]{#4}\end{minipage}
    \if\relax\detokenize{#1}\relax\else\nestfigcaption{#1}{#5}\fi
    \vspace{-8pt}
  \end{figure}%
}
\newcommand{\nestfigloadcaption}[5][]{%
  \begin{figure}[H]
    \centering
    \begin{minipage}[t]{0.31\linewidth}\centering\nestincludefig[width=\linewidth]{#2}\end{minipage}\hfill
    \begin{minipage}[t]{0.31\linewidth}\centering\nestincludefig[width=\linewidth]{#3}\end{minipage}\hfill
    \begin{minipage}[t]{0.31\linewidth}\centering\nestincludefig[width=\linewidth]{#4}\end{minipage}
    \if\relax\detokenize{#1}\relax\else\nestfigcaption{#1}{#5}\fi
  \end{figure}%
}

\title{\textbf{Nested Episodic State Topology (NEST): A Graph-Theoretic Architecture of Cognitive States}}
\author{Ishant}
\date{}

\begin{document}
\maketitle

\begin{abstract}
We present NEST (Nested Episodic State Topology), a foundational graph-theoretic representational ontology for modeling cognition as structured state formation and transformation rather than as a finished empirical model. Concepts, episodes, percepts, and task contexts are represented as typed, weighted graphs whose nodes may carry internal subgraph payloads; edges are typed under six relation classes---causal, containment, temporal, associative, evidential, and spatial. Durable belief graphs are separated from capacity-limited working-memory graphs that may host transient non-belief content. WM--belief grounding, conflict catalogs, and belief-update operators specify how transient structure is tested against stored knowledge and how belief is revised. A reusable operator toolkit---activation, graph-property functionals, working-memory transitions, awareness and trajectory functionals, and belief update---organizes the formal core. Derived diagnostics such as fragmentation, involvement, signed evaluation, coherence, and active conflict define familiar phenomena in the same ontology; self-related processing is modeled through designated self-image subgraphs within belief. Subsequent sections instantiate this core without new primitives: phenomena signatures, a task-instantiation schema for action selection and failure modes, and compatibility mappings that embed ACT-R, Soar, Sigma, the Common Model of Cognition, Global Workspace Theory, semantic networks, Theory-Theory, and chunking as constrained regions of one language. Mappings constitute the culminating technical section; discussion addresses scope, limitations, and open research directions. The contribution is intentionally foundational: a transparent representational substrate for later empirical, computational, and domain-specific work.
\end{abstract}

\section{Introduction}

Theoretical cognitive science has long been marked by fragmentation.
The field has produced many successful local theories of memory, reasoning, perception, language, control, and learning, but it still lacks a common representational language in which these theories can be stated, compared, and revised on equal footing \cite{Newell1990,Milkowski2017}.
This makes integration difficult: two theories may appear incompatible simply because they are expressed in different formal systems, even when they are making structurally similar claims.
The central problem, then, is not only how cognition works, but how cognitive science should represent what its theories are about.

Many integrative architectures unify cognition by committing to a particular processing mechanism, such as a production system, a workspace, or a control loop.
That approach is valuable, but it often forces translation into one architecture's native procedural vocabulary before theories can be compared.
What is missing is a representational layer prior to mechanism: a shared substrate in which theories can be expressed as structural commitments, so that comparison turns on what cognitive structure they assume rather than on how they happen to be implemented.

NEST (Nested Episodic State Topology) is motivated by this gap.
It proposes a graph-theoretic representational ontology in which cognition is modeled as structured state formation and transformation.
Cognitive states are nested graphs whose nodes can carry internal structure, including subgraphs, so that concepts, episodes, percepts, and task contexts can be represented in one formal language while preserving typed relations among causal, temporal, spatial, associative, evidential, and containment structure.

Human cognition must also distinguish transient content from durable knowledge---tentative hypotheses, unresolved conflict, and stored beliefs that guide interpretation and action.
A single undifferentiated memory store does not capture these distinctions well enough for theoretical comparison.
NEST therefore separates durable belief graphs from capacity-limited working-memory graphs, allowing perceptual input and intermediate reasoning to be represented without immediately collapsing them into stored knowledge.

Our contributions are fourfold.
First, we specify the recursive node-and-relation ontology and the machinery linking working memory and belief.
Second, we define graph-theoretic signatures for key cognitive phenomena in the same ontology.
Third, we introduce a task-instantiation schema for action selection, outcome prediction, and failure-mode reasoning.
Fourth, we show how major frameworks in cognitive science can be expressed as special cases or constrained regions of the NEST ontology, with compatibility mappings as the culminating technical step.
The paper is foundational rather than a finished empirical model: its aim is theoretical comparison, representational transparency, and a substrate for later domain-specific development.

The paper proceeds as follows.
Section~\ref{sec:formalism} introduces the recursive ontology and operator toolkit.
Section~\ref{sec:phenomena} derives graph-theoretic signatures for key cognitive phenomena.
Section~\ref{sec:task-instantiation} presents the task-instantiation schema for action selection, outcome prediction, and failure-mode reasoning.
Section~\ref{sec:mappings} develops compatibility mappings to major cognitive frameworks.
Section~\ref{sec:discussion} closes with implications, limitations, and open research directions.

\section{Formal Architecture}
\label{sec:formalism}

We now specify the formal core of NEST. We begin with basic graph-theoretic definitions, introduce a small reusable library of generic functionals and operators, and then instantiate that library for long-term belief graphs, working-memory graphs, awareness, control, and graph-theoretic phenomena. Throughout, most cognitive constructs are treated as instances of activation functions, graph-property functionals, belief-update operators, or trajectory functionals over working-memory structure.

\subsection{Graphs and recursive nodes}

\begin{definition}[Base graph]
\label{def:base-graph}
A \emph{base graph} is a tuple
\[
G = (V, E, \mathcal{R}, \tau, w)
\]
where \(V\) is a finite set of nodes, \(E \subseteq V \times V\) is a set of directed edges, \(\mathcal{R}\) is a set of relation types, \(\tau : E \rightarrow \mathcal{R}\) assigns a type to each edge, and \(w : E \rightarrow \mathbb{R}\) assigns a real-valued weight to each edge.
By default, non-evidential edge classes use non-negative weights; signed weights are reserved for edge types whose semantics require polarity, such as evidential support and contradiction.
\end{definition}

\nestfigsmall[Base graph $G=(V,E,\mathcal{R},\tau,w)$: nodes $v_i\in V$, directed edges $(v_i,v_j)\in E$, and edge labels $r_k\in\mathcal{R}$ illustrating $\tau(e)$ for representative edges $e$ (specific relation names are domain-dependent and belong in captions or instantiations, not in schematic figures).]{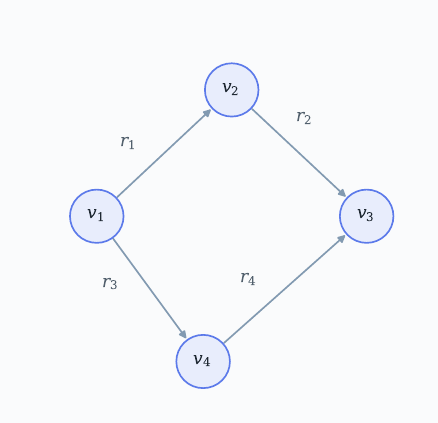}

Figure~\ref{fig:fig01_base_graph} illustrates a representative base graph with typed, weighted edges.

\begin{definition}[Representational payload]
\label{def:representational-payload}
A \emph{representational payload} is an internal structure associated with a node.
For a node \(v\), its payload \(P_v\) may take one of the following forms:
\begin{enumerate}
    \item an empty payload \(\varnothing\), corresponding to an atomic node,
    \item a linguistic payload, such as a word, phrase, sentence, or symbolic chunk,
    \item a perceptual payload, such as a visual, auditory, tactile, or motor chunk,
    \item a structured perceptual payload, i.e., a perceptual payload whose internal representation is a subgraph encoding a perceptual scene, feature bundle, or sensorimotor pattern,
    \item a graph payload \(G_v\), where \(G_v\) is itself a base graph.
\end{enumerate}
\end{definition}

\begin{definition}[Recursive node]
\label{def:recursive-node}
A \emph{recursive node} is a tuple
\[
v = (c_v, P_v),
\]
where \(c_v\) is a node label or identifier and \(P_v\) is a representational payload.
If \(P_v\) is a structured perceptual payload or a graph payload, then \(v\) is said to be \emph{internally structured}, and its internal subgraph may encode the node's perceptual scene, conceptual organization, episodic content, or procedural structure.
\end{definition}

\nestfig[Recursive node $v=(c_v,P_v)$ with the five representational payload types from Definition~\ref{def:representational-payload}: (1)~empty $P_v=\varnothing$; (2)~linguistic $P_v$ (e.g.\ a word token $\langle w\rangle$); (3)~atomic perceptual $P_v$; (4)~structured perceptual $P_v$ as an internal subgraph; (5)~graph payload $G_v$ with internal nodes $u_i$.]{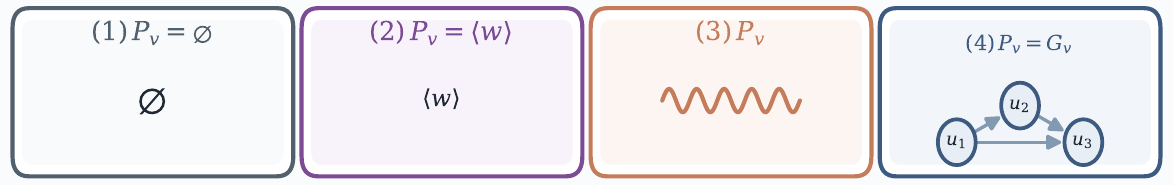}

Figure~\ref{fig:fig02_recursive_node} summarizes the five payload types in Definition~\ref{def:representational-payload}.

This definition allows NEST nodes to function at multiple levels of abstraction.
A node may denote a simple concept, a linguistic token, an atomic or structured perceptual chunk, or a recursively structured subgraph.
For example, a node for a person may contain an internal role graph, a node for an object may contain a visual feature graph, and a node for a spoken phrase may contain an auditory or phonological chunk.

We write \(v \in V\) for nodes in a base graph. When \(P_v\) is a graph payload \(G_v\), or a structured perceptual payload whose internal representation is a subgraph \(G_v\), we may emphasize the internal structure by writing \(v = (c(v), G_v)\).

\begin{definition}[Edge type]
\label{def:edge-type-taxonomy}
Let \(G = (V, E, \mathcal{R}, \tau, w)\) be a base graph (Definition~\ref{def:base-graph}), where each node \(v \in V\) carries a representational payload \(P_v\) as in Definition~\ref{def:representational-payload}. The relation-type set \(\mathcal{R}\) partitions into six pairwise-disjoint classes:
\[
\mathcal{R} = \mathcal{R}_{\mathrm{causal}} \cup \mathcal{R}_{\mathrm{contain}} \cup \mathcal{R}_{\mathrm{temp}} \cup \mathcal{R}_{\mathrm{assoc}} \cup \mathcal{R}_{\mathrm{epist}} \cup \mathcal{R}_{\mathrm{spatial}},
\]
each defined by an independent structural test below. For each edge \(e = (u, v) \in E\), the assigned type \(\tau(e) \in \mathcal{R}\) belongs to exactly one class. A particular graph may employ only a subset; belief-graph and working-memory instantiations specify which classes appear in \(\mathcal{R}_B\) and \(\mathcal{R}_{W_t}\). The six classes are:
\begin{enumerate}
    \item \textbf{Causal} (\(\mathcal{R}_{\mathrm{causal}}\)). An edge \((u,v)\) with \(\tau(u,v) \in \mathcal{R}_{\mathrm{causal}}\) asserts that node \(u\) produces, enables, probabilistically raises, or inhibits the occurrence of node \(v\), independent of the payload types of \(u\) and \(v\):
    \[
    \mathcal{R}_{\mathrm{causal}} = \{r_{\mathrm{causes}},\ r_{\mathrm{enables}},\ r_{\mathrm{inhibits}},\ r_{\mathrm{counterfactual}}\}.
    \]
    Causal edges are antisymmetric: \((u,v) \in E_{\mathcal{R}_{\mathrm{causal}}} \Rightarrow (v,u) \notin E_{\mathcal{R}_{\mathrm{causal}}}\). Acyclicity over \(\mathcal{R}_{\mathrm{causal}}\) is a structural forbidden pattern in \(\mathcal{S}\). When \(u\) is a distinguished agent node, an optional \emph{attribution tag} \(a(u,v) \in V\) records whose causal model the edge belongs to---e.g., ``agent \(a\) believes \(u\) causes \(v\)''.

    \item \textbf{Containment} (\(\mathcal{R}_{\mathrm{contain}}\)). An edge \((u,v)\) with \(\tau(u,v) \in \mathcal{R}_{\mathrm{contain}}\) asserts that \(v\)'s existence or identity is nested inside \(u\)'s scope, characterized by an \emph{existential-dependency test}: \((u,v) \in E_{\mathcal{R}_{\mathrm{contain}}}\) iff removing \(u\) from \(G\) requires removing \(v\) as well.
    \[
    \mathcal{R}_{\mathrm{contain}} = \{r_{\mathrm{is\mathrm{-}a}},\ r_{\mathrm{part\mathrm{-}of}},\ r_{\mathrm{instance\mathrm{-}of}},\ r_{\mathrm{role\mathrm{-}of}}\}.
    \]
    When \(v\) carries a graph payload \(G_v\) (Definition~\ref{def:representational-payload}, item~5) or a structured perceptual payload (item~4), its internal subgraph is linked to \(v\) via \(r_{\mathrm{part\mathrm{-}of}}\) edges. Containment is transitive and irreflexive; a directed cycle over \(\mathcal{R}_{\mathrm{contain}}\) is a forbidden pattern in \(\mathcal{S}\).

    \item \textbf{Temporal} (\(\mathcal{R}_{\mathrm{temp}}\)). An edge \((u,v)\) with \(\tau(u,v) \in \mathcal{R}_{\mathrm{temp}}\) encodes ordering of nodes across time via a node timestamp function \(\theta_V : V \to \mathbb{R}\) (distinct from the working-memory time index \(t\)) such that \((u,v) \in E_{\mathcal{R}_{\mathrm{temp}}} \Rightarrow \theta_V(u) \leq \theta_V(v)\):
    \[
    \mathcal{R}_{\mathrm{temp}} = \{r_{\mathrm{before}},\ r_{\mathrm{after}},\ r_{\mathrm{during}},\ r_{\mathrm{next}}\}.
    \]
    Here \(r_{\mathrm{before}}\) and \(r_{\mathrm{after}}\) are mutually exclusive on the same ordered pair (a predicate incompatibility in \(\mathcal{F}\)); \(r_{\mathrm{next}}\) is the sequential successor used in action-sequence and planning subgraphs.

    \item \textbf{Associative} (\(\mathcal{R}_{\mathrm{assoc}}\)). An edge \((u,v)\) with \(\tau(u,v) \in \mathcal{R}_{\mathrm{assoc}}\) encodes statistical co-occurrence or spreading-activation linkage with no stronger typed commitment \cite{CollinsQuillian1969,CollinsLoftus1975}:
    \[
    \mathcal{R}_{\mathrm{assoc}} = \{r_{\mathrm{related\mathrm{-}to}},\ r_{\mathrm{co\mathrm{-}occurs\mathrm{-}with}},\ r_{\mathrm{similar\mathrm{-}to}}\}.
    \]
    Associative edges are undirected, represented as symmetric pairs \((u,v),(v,u) \in E\) with equal weights \(w(u,v) = w(v,u)\). This is the default class: any relation not warranting a stronger typed commitment from classes 1--3 or 5--6 defaults here.

    \item \textbf{Evidential} (\(\mathcal{R}_{\mathrm{epist}}\)). An edge \((u,v)\) with \(\tau(u,v) \in \mathcal{R}_{\mathrm{epist}}\) encodes inferential, evidential, or evaluative support between propositional or evaluative nodes:
    \[
    \mathcal{R}_{\mathrm{epist}} = \{r_{\mathrm{supports}},\ r_{\mathrm{contradicts}},\ r_{\mathrm{derived\mathrm{-}from}},\ r_{\mathrm{supersedes}}\}.
    \]
    Valence and deontic force are carried as \emph{signed edge weights} on \(r_{\mathrm{supports}}\) and \(r_{\mathrm{contradicts}}\) edges (extending \(w\) to \(\mathbb{R}\) on \(\mathcal{R}_{\mathrm{epist}}\) when needed); an obligation is an \(r_{\mathrm{supports}}\) edge with \(|w(u,v)| = 1\) directed at a norm node, and a prohibition is the corresponding \(r_{\mathrm{contradicts}}\) edge. Here \(r_{\mathrm{contradicts}}\) induces a conflict pair in \((\mathcal{F}, \mathcal{S})\); \(r_{\mathrm{supersedes}}\) is generated by belief-update operators when a node is revised rather than deleted.

    \item \textbf{Spatial} (\(\mathcal{R}_{\mathrm{spatial}}\)). An edge \((u,v)\) with \(\tau(u,v) \in \mathcal{R}_{\mathrm{spatial}}\) encodes geometric relations, restricted to node pairs where at least one endpoint carries a perceptual or structured perceptual payload (Definition~\ref{def:representational-payload}, items~3--4):
    \[
    \mathcal{R}_{\mathrm{spatial}} = \{r_{\mathrm{above}},\ r_{\mathrm{below}},\ r_{\mathrm{left\mathrm{-}of}},\ r_{\mathrm{right\mathrm{-}of}},\ r_{\mathrm{near}},\ r_{\mathrm{affords}}\}.
    \]
    The sub-type \(r_{\mathrm{affords}}\) links a body-schema node \(u\) to a perceptual node \(v\), encoding a sensorimotor affordance relation.
\end{enumerate}
For any edge class \(\mathcal{R}_k\), write \(E_{\mathcal{R}_k} = \{(u,v) \in E \mid \tau(u,v) \in \mathcal{R}_k\}\) for the induced edge set. For a subgraph \(Q \subseteq B\), the restriction \(E_{\mathcal{R}_k}(Q) = E_{\mathcal{R}_k} \cap (V_Q \times V_Q)\) denotes the edges of that class internal to \(Q\).
\end{definition}

\begin{remark}[Identification and attributed mental states]
\label{rem:edge-type-derived}
Coreference and attributed mental states are expressed within the six edge classes of Definition~\ref{def:edge-type-taxonomy}. Coreference is not represented by mutual containment, since containment is irreflexive and acyclic. Instead, coreference is represented either by a distinguished evidential or associative identification relation, or by quotienting an explicitly declared coreference relation external to \(\mathcal{R}_{\mathrm{contain}}\). Informally, two nodes are coreferential when they are declared to denote the same referent and may be collapsed into a single equivalence class for selected operations. Attributed beliefs and trust are causal or evidential edges carrying an attribution tag \(a(u,v) \in V\): an agent's belief that \(u\) causes \(v\) is a causal edge tagged with that agent's node, and trust in a source is a positively weighted evidential edge from the agent to that source. Recognition events may promote a partial match to an explicit coreference or identification link in \(B\).
For each socially relevant agent \(i\), an \emph{other-agent anchor node} \(v_{\mathrm{self}}^{(i)} \in V_B\) is a designated role for a belief-graph node, not a new primitive node type. An \emph{attributed-agent subgraph} \(A_t^{(i)} \subseteq B\) is centered on \(v_{\mathrm{self}}^{(i)}\) and contains beliefs, goals, preferences, and predicted actions attributed to that agent; its edges are causal or evidential edges as already permitted above.
\label{def:attributed-agent-subgraph}
\end{remark}

\begin{remark}[Cross-class edge assignments]
\label{rem:edge-type-cross-class}
The partition \(\mathcal{R} = \bigcup_k \mathcal{R}_k\) does not preclude multiple edges of different types between the same node pair---e.g., \((u,v)\) may simultaneously carry \(r_{\mathrm{causes}}\) and \(r_{\mathrm{before}}\). The incompatibility constraint \((\mathcal{F}, \mathcal{S})\) (Definition~\ref{def:incompatibility}) does not require \(\mathcal{F}\) or \(\mathcal{S}\) to respect the class partition: cross-class coexistence on the same ordered pair is permitted by default. Belief-graph constraints \((\mathcal{F}_B, \mathcal{S}_B)\) are declared primarily to forbid \emph{within-class} contradictions; see Remark~\ref{rem:incompatibility-examples}.
\end{remark}

\begin{remark}[Edge type summary]
\label{rem:edge-type-summary}
The six canonical edge-type classes are summarized below.
\begin{center}
\small
\begin{tabular}{@{}p{0.14\linewidth}p{0.10\linewidth}p{0.30\linewidth}p{0.38\linewidth}@{}}
\hline
Class & Symbol & Key sub-types & Structural test \\
\hline
Causal & \(\mathcal{R}_{\mathrm{causal}}\) & causes, enables, inhibits, counterfactual & Antisymmetry; acyclicity \\
Containment & \(\mathcal{R}_{\mathrm{contain}}\) & is-a, part-of, instance-of, role-of & Existential dependency; transitivity \\
Temporal & \(\mathcal{R}_{\mathrm{temp}}\) & before, after, during, next & \(\theta_V(u) \leq \theta_V(v)\) \\
Associative & \(\mathcal{R}_{\mathrm{assoc}}\) & related-to, co-occurs-with, similar-to & Symmetry; default class \\
Evidential & \(\mathcal{R}_{\mathrm{epist}}\) & supports, contradicts, supersedes & Signed weight; conflict pairs \\
Spatial & \(\mathcal{R}_{\mathrm{spatial}}\) & above, left-of, affords & Perceptual payload restriction \\
\hline
\end{tabular}
\end{center}
\end{remark}

Definition~\ref{def:edge-type-taxonomy} supplies the taxonomy for typed edges in any base graph.

\subsection{Belief graph and mental models}

\begin{definition}[Belief graph]
A \emph{belief graph} is a base graph
\[
B = (V_B, E_B, \mathcal{R}_B, \tau_B, w_B),
\]
whose nodes are recursive nodes representing concepts, episodes, and other structured knowledge. \(B\) is intended to approximate long-term memory and background knowledge.
\end{definition}

\begin{definition}[Mental model]
\label{def:mental-model}
A \emph{mental model} is a connected subgraph
\[
M = (V_M, E_M) \subseteq B
\]
such that the nodes and edges in \(M\) jointly represent a particular situation, domain, or system (e.g., a physical scenario, a social interaction, a mathematical problem).
\end{definition}

\nestfigsmallw[A mental model \(M\) as a connected subgraph of the belief graph \(B\).]{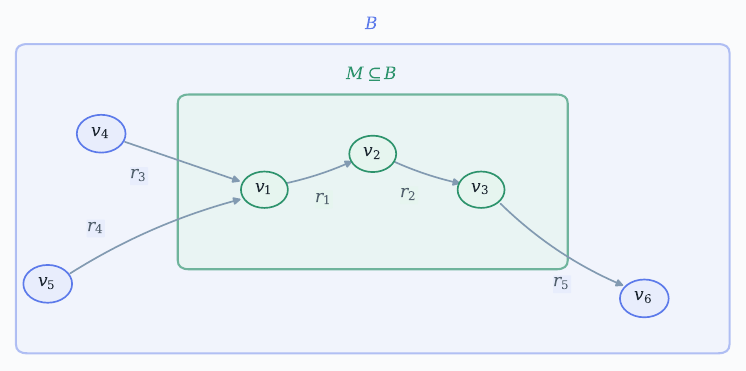}{0.552}

Figure~\ref{fig:fig03_belief_mental_model} shows a connected subgraph \(M \subseteq B\).

Mental models are thus structured subgraphs of \(B\). Multiple mental models may overlap in nodes and edges.

\subsection{Working-memory graphs}

\begin{definition}[Working-memory graph]
A \emph{working-memory graph} at time \(t\), denoted \(W_t\), is a base graph
\[
W_t = (V_{W_t}, E_{W_t}, \mathcal{R}_{W_t}, \tau_{W_t}, w_t)
\]
subject to capacity constraints on nodes and edges (e.g., \(|V_{W_t}| \leq C_V\), \(|E_{W_t}| \leq C_E\), where \(C_V, C_E\) are architecture-specific constants).
Most nodes in \(V_{W_t}\) are linked, via one or more edges, to nodes in the belief graph \(B\), but \(V_{W_t}\) may also contain transient nodes that encode sensory, hypothetical, or not-yet-integrated content.
The weight function \(w_t\) may differ from \(w_B\) to reflect transient activation strengths.
\end{definition}

\begin{definition}[Belief-anchored working memory]
A working-memory state \(W_t\) is \emph{belief-anchored} if every node in \(V_{W_t}\) corresponds to some node in \(V_B\) and every edge in \(E_{W_t}\) corresponds to an edge in \(E_B\).
In this case we write \(W_t \sqsubseteq B\) to emphasize that \(W_t\) is realized as a subgraph of the belief graph.
\end{definition}

\begin{definition}[Perceptual input node]
\label{def:perceptual-input-node}
A \emph{perceptual input node} at time \(t\) is a node
\[
v^{\mathrm{perc}} \in V_{W_t} \setminus V_B
\]
such that its payload \(P_{v^{\mathrm{perc}}}\) is a perceptual payload, and either
\[
P_{v^{\mathrm{perc}}} \in \mathcal{P}_{\mathrm{atom}}
\quad\text{or}\quad
P_{v^{\mathrm{perc}}} = G_{\mathrm{perc}} = (V_{\mathrm{perc}}, E_{\mathrm{perc}}, \mathcal{R}_{\mathrm{perc}}, \tau_{\mathrm{perc}}, w_{\mathrm{perc}}),
\]
where \(\mathcal{P}_{\mathrm{atom}}\) is the set of atomic perceptual payloads (Definition~\ref{def:representational-payload}, item~3) and \(G_{\mathrm{perc}}\) is a connected base graph encoding a perceptual scene, feature bundle, or sensorimotor pattern, whose edges draw on spatial, temporal, and property relations from the edge-type taxonomy (Definition~\ref{def:edge-type-taxonomy}, classes \(\mathcal{R}_{\mathrm{spatial}}\) and \(\mathcal{R}_{\mathrm{temp}}\)).

Its activation \(\alpha_t^{W}(v^{\mathrm{perc}})\) is initialized by input strength and evolves under the working-memory transition operator \(U\) unless the node is integrated into \(B\).

The set of perceptual input nodes active at time \(t\) is
\[
V^{\mathrm{perc}}_t = \bigl\{ v \in V_{W_t} \setminus V_B : \bigl(P_v \in \mathcal{P}_{\mathrm{atom}} \lor (\exists\, G_{\mathrm{perc}} \text{ connected},\, P_v = G_{\mathrm{perc}})\bigr) \land \alpha_t^{W}(v) \geq \theta_W \bigr\}.
\]
\end{definition}

Integrating a perceptual input node or chunk from working memory into the belief graph \(B\) is termed \textbf{encoding}: a micro-scale revision that adds the transient node to \(V_B\), grounds edges from \(E_{W_t}\) into \(E_B\), and may align structure with a stored mental model when graph similarity exceeds an architecture-specific threshold (Remark~\ref{rem:encoding-gate}).
Repeated micro-applications over a trajectory \((W_t, \dots, W_{t+n})\) accumulate into a new mental model \(M' \subseteq B'\). Macro-scale compression of recently encoded structure into a single high-level recursive node---\emph{schema consolidation}---reorganizes belief-graph structure at larger scale.

\begin{remark}[Perceptual scene organization]
\label{rem:perceptual-scene-organization}
A currently perceived environment may be represented by a designated subgraph \(S_t^{\mathrm{scene}} \subseteq W_t\) whose nodes are perceptual input nodes or chunked perceptual structures and whose edges draw on spatial, temporal, and containment relations (Definition~\ref{def:edge-type-taxonomy}). This is an interpretation-level specialization of perceptual-input and recursive-node machinery, not a new primitive graph type.
\end{remark}

\begin{remark}[Body-schema subgraph]
\label{rem:body-schema-subgraph}
A \emph{body-schema subgraph} \(B^{\mathrm{body}}=(V_{\mathrm{body}},E_{\mathrm{body}}) \subseteq B\) encodes the agent's model of its own body, proprioceptive sensations, and motor capabilities, with spatial and affordance edges as in \(\mathcal{R}_{\mathrm{spatial}}\) (Definition~\ref{def:edge-type-taxonomy}). Formally, it is a designated belief subgraph selected by functional role rather than a new primitive graph type.
\end{remark}

Working memory often has the same schematic form as the mental model in Figure~\ref{fig:fig03_belief_mental_model}: a relatively small, capacity-limited graph whose nodes are mostly drawn from, and connected back to, the larger belief graph \(B\).
Capacity limits \((C_V, C_E)\), the time-varying weight function \(w_t\), and the possibility of transient non-belief nodes distinguish \(W_t\) from a stored mental model \(M\).
We therefore omit a separate figure.

\subsection{Generic functionals, operators, and constraints}
\label{sec:generic-functionals}

Before introducing the generic toolkit, we fix the role of the main symbol families used throughout this section. Activation maps \(\alpha\) govern node salience and WM membership; graph-property functionals \(\Phi\) return scalar or vector diagnostics of working-memory structure; conflict catalogs \(\mathfrak{C}\) enumerate incompatible region pairs; coherence and involvement functionals \(\mathcal{H}\) and \(\mathcal{I}\) summarize conflict and participation over designated subgraphs; transition operators \(U\) evolve working-memory states; belief-update operators \(\Delta\) revise \(B\) from finite WM trajectories; awareness \(\mathcal{A}\) selects the accessed subgraph of \(W_t\); and trajectory functionals \(\Psi\), including the control functional \(\mathcal{K}\), summarize or regulate temporal evolution.

Most constructs in the formal architecture---graph-defined properties such as subgraph involvement, signed evaluation, coherence, graph-property descriptors, belief update, and trajectory-based regulation---are instances of the generic function types defined here.

\begin{definition}[Activation maps]
\label{def:activation-function}
An \emph{belief-node activation function} at time \(t\) is a mapping
\[
\alpha_t^B : V_B \rightarrow \mathbb{R}_{\ge 0}.
\]
For a subgraph \(X \subseteq B\), we write
\[
\alpha_t^B(X) = \sum_{v \in V_X} \alpha_t^B(v)
\]
for its total activation.
When transient nodes are present in working memory, \(\alpha_t^B\) is extended to a \emph{working-memory activation map}
\[
\alpha_t^{W} : V_{W_t} \rightarrow \mathbb{R}_{\ge 0},
\]
such that \(\alpha_t^{W}(v)=\alpha_t^{B}(v)\) for \(v \in V_B \cap V_{W_t}\), and \(\alpha_t^{W}(v)\) is initialized by input strength for \(v \in V_{W_t}\setminus V_B\), decaying under \(U\) unless the node is integrated into \(B\).
For a subgraph \(X \subseteq W_t\), we write \(\alpha_t^{W}(X) = \sum_{v \in V_X} \alpha_t^{W}(v)\).
Working-memory membership is defined by thresholding against an architecture-specific cutoff \(\theta_W \geq 0\): for \(v \in V_B\), we say that \(v\) is \emph{belief-active in working memory} at \(t\), written \(\mathrm{active}_B(v,W_t)\), when \(\alpha_t^B(v) \geq \theta_W\).
Accordingly, the belief-derived portion of working memory is
\[
V_{W_t}^{B}=\{v \in V_B \cap V_{W_t} : \mathrm{active}_B(v,W_t)\}.
\]
The transient portion is \(V_{W_t}^{\mathrm{tr}} = V_{W_t}\setminus V_B\); a transient node \(v\) is WM-active when \(\alpha_t^{W}(v) \geq \theta_W\).
For \(X \subseteq V_B\), the \emph{active restriction} is
\[
\mathrm{active}(X,W_t)=X\cap V_{W_t}^{B}.
\]
\end{definition}

\begin{definition}[Salience set]
\label{def:salience-set}
A \emph{salience set} at time \(t\) is any designated subset \(X_t \subseteq V_B\).
Metrics over designated node-sets are written in terms of \(\mathrm{active}(X_t, W_t)\).
\end{definition}

\begin{definition}[Node-set satisfaction]
\label{def:node-set-satisfaction}
The \emph{node-set satisfaction} of a finite nonempty salience set \(X_t\) is
\[
S_{\mathrm{sat}}(X_t, W_t) = \frac{|\mathrm{active}(X_t, W_t)|}{|X_t|}.
\]
\end{definition}

\begin{definition}[Graph property functional]
\label{def:graph-property-functional}
A \emph{graph property functional} is any mapping
\[
\Phi : \mathcal{W} \rightarrow \mathcal{D},
\]
where \(\mathcal{D}\) is a descriptor space of scalar or finite tuple outputs. In the present paper, \(\mathcal{D}\) is usually instantiated as \(\mathbb{R}^k\). Individual components are written \(\Phi_j(W_t)\) or with descriptive subscripts; domain-specific components are introduced where the corresponding domain is defined.
\end{definition}

\subsubsection{Incompatibility and conflict}
\label{sec:incompatibility-conflict}

A graph may declare forbidden label combinations and relational patterns.
Working memory and belief are linked by a collation graph when transient
content must be tested against stored structure. The definitions below
specify a single detection rule, a derived conflict catalog, and derived
metrics that summarize active conflict and coherence.

\begin{definition}[Incompatibility constraint]
\label{def:incompatibility}
Let \(G = (V, E, \mathcal{R}, \tau, w)\) be a base graph. An
\emph{incompatibility constraint} on \(G\) is a pair
\((\mathcal{F}, \mathcal{S})\), where:
\begin{enumerate}
    \item \textbf{Predicate incompatibility} \(\mathcal{F}\).
    \(\mathcal{F}\) is a set of label pairs \(\{\ell_1,\ell_2\}\),
    where each \(\ell_i\) is either a node label or an edge type in
    \(\mathcal{R}\).     Each pair \(\{\ell_1, \ell_2\}\) specifies that the
    two labels may not simultaneously apply to the same referent. If
    \(\ell_1\) and \(\ell_2\) are node labels, they may not both hold of
    the same node \(v \in V\); if they are edge types, they may not both
    hold of the same ordered pair \((u,v) \in V \times V\). A subgraph
    \(G' \subseteq G\) \emph{violates} \(\mathcal{F}\) if there exists
    \(\{\ell_1, \ell_2\} \in \mathcal{F}\) such that a single referent in
    \(G'\) carries both labels.

    \item \textbf{Relational incompatibility} \(\mathcal{S}\).
    \(\mathcal{S}\) is a set of forbidden pattern schemas
    \[
    P = (V_P, E_P, \mathcal{R}_P, \tau_P, \omega_P),
    \]
    where \(\omega_P\) is an optional predicate imposing constraints on
    the weights of matched edges. A subgraph \(G' \subseteq G\)
    \emph{matches} \(P\) if there exists a label-preserving isomorphism
    \(f : V_P \to V_{G'}\) such that the induced correspondence preserves
    edges and \(\tau_P\)-labels, and such that \(\omega_P\) is satisfied
    under the weight function \(w\).
\end{enumerate}
Two subgraphs \(X\) and \(Y\) of \(G\) are \emph{incompatible}, written
\(X \bowtie Y\), relative to \((\mathcal{F}, \mathcal{S})\), iff
\(X \cup Y\) violates \(\mathcal{F}\) or matches some
\(P \in \mathcal{S}\). A subgraph \(X\) is \emph{internally incompatible}
if \(X \bowtie X\). The graph \(G\) is \emph{well-formed} under
\((\mathcal{F}, \mathcal{S})\) iff \(G\) is not internally incompatible.
The sets \(\mathcal{F}\) and \(\mathcal{S}\) are assumed to be finite,
or to belong to finitely parameterized families of forbidden patterns, so
that incompatibility checking is decidable. For the belief graph \(B\),
the relation-type set \(\mathcal{R}_B\) is equipped with a declared
constraint \((\mathcal{F}_B, \mathcal{S}_B)\), which specifies the
combinations and patterns that are ill-formed in \(B\) and in any
working-memory graph derived from it.
\end{definition}

\begin{remark}[Examples]
\label{rem:incompatibility-examples}
Predicate incompatibilities include \((\texttt{alive}, \texttt{dead})\)
on the same node, \((r_{\mathrm{causes}}, r_{\mathrm{inhibits}})\) and
\((r_{\mathrm{before}}, r_{\mathrm{after}})\) on the same ordered
node-pair. Structural forbidden patterns include a directed cycle over
\(r_{\mathrm{is\mathrm{-}a}}\), an acyclicity violation over
\(\mathcal{R}_{\mathrm{causal}}\), or opposing signed
\(r_{\mathrm{supports}}\) and \(r_{\mathrm{contradicts}}\) edges on the
same target with incompatible magnitudes via a pattern \(P \in
\mathcal{S}\) whose weight predicate \(\omega_P\) requires
\(w(e_1)\,w(e_2) < 0\) on the matched edges.
\end{remark}

\begin{definition}[WM--belief collation]
\label{def:wm-belief-collation}
Let \(W_t\) be a working-memory graph and \(B\) the belief graph. Write
\(V^{\mathrm{tr}}_t = V_{W_t} \setminus V_B\) for transient WM nodes.
A \emph{grounding correspondence} at time \(t\) is a finite relation
\(\Gamma_t \subseteq V^{\mathrm{tr}}_t \times V_B\) with scores
\(\sigma_t(v,u) \in [0,1]\) measuring similarity between transient node
\(v\) and belief candidate \(u\). Let
\[
\Gamma_t^{\geq\theta_{\mathrm{gnd}}} =
\bigl\{ (v,u) \in \Gamma_t : \sigma_t(v,u) \geq \theta_{\mathrm{gnd}} \bigr\}.
\]
Several transient nodes may ground to the same belief node; a transient node
may have multiple qualifying candidates.
The \emph{WM--belief collation graph} is
\[
G^{\mathrm{col}}_t = \mathrm{Coll}(W_t, B, \Gamma_t),
\]
with node set \(V^{\mathrm{col}}_t = V_B \cup V^{\mathrm{tr}}_t\). It
contains all edges of \(E_B\) and \(E_{W_t}\) whose endpoints lie in
\(V^{\mathrm{col}}_t\), plus an associative grounding edge \((v,u)\) of type
\(r_{\mathrm{related\mathrm{-}to}}\) for each \((v,u) \in \Gamma_t^{\geq\theta_{\mathrm{gnd}}}\).
The graph \(G^{\mathrm{col}}_t\) is \emph{collation-well-formed} when it is
not internally incompatible under \((\mathcal{F}_B, \mathcal{S}_B)\)
(Definition~\ref{def:incompatibility}).
\end{definition}

\begin{definition}[Cross-memory incompatibility]
\label{def:cross-memory-incompatibility}
Let \((\mathcal{F}_B, \mathcal{S}_B)\) be the belief-graph incompatibility
constraint. Regions \(X\) and \(Y\) are \emph{incompatible}, written
\(X \bowtie Y\), relative to \((W_t, B)\) and \((\mathcal{F}_B, \mathcal{S}_B)\),
as follows.
Single-graph regions are tested directly by Definition~\ref{def:incompatibility}.
WM--belief regions are tested in the collation graph \(G^{\mathrm{col}}_t\) (Definition~\ref{def:wm-belief-collation}), and incompatible whenever the induced WM--belief subgraph violates \((\mathcal{F}_B,\mathcal{S}_B)\) or a forbidden schema cross-matches a transient payload region.
For a transient node \(v \in V^{\mathrm{tr}}_t\), the \emph{payload region} is \(V_{\mathrm{pl}}(v) = V_{G_v}\) when \(v\) carries a graph payload \(G_v\) and \(V_{\mathrm{pl}}(v) = \{v\}\) otherwise; a forbidden schema \(P \in \mathcal{S}_B\) \emph{cross-matches} \((v, Y)\) when \(P\) matches the subgraph induced by \(V_{\mathrm{pl}}(v) \cup V_Y\) together with relevant grounding edges \((v,u) \in \Gamma_t^{\geq\theta_{\mathrm{gnd}}}\).
\end{definition}

\begin{remark}[Belief-anchored working memory]
\label{rem:belief-anchored-wm-incompatibility}
When \(W_t \sqsubseteq B\), the WM--belief case of
Definition~\ref{def:cross-memory-incompatibility} reduces to the
single-graph test on \(B\).
\end{remark}

\begin{proposition}[WM--belief conflict catalog]
\label{prop:wm-belief-conflict-catalog}
Given \(W_t\), \(B\), and \(\Gamma_t\), the \emph{WM--belief conflict
catalog}
\[
\mathfrak{C}_{\mathrm{wb}}(W_t, B) =
\bigl\{ (X,Y) : X \subseteq W_t,\; Y \subseteq B,\;
V_Y \cap V_{W_t} = \varnothing,\; X \bowtie Y \bigr\},
\]
is the set of incompatible WM--belief region pairs at \(t\).
\end{proposition}

\begin{remark}[Conflict classification]
\label{rem:conflict-classification}
Catalogued pairs \(X \bowtie Y\) are classified by where the incompatible
regions lie:
\begin{enumerate}
    \item \textbf{Intra-WM}: \(X, Y \subseteq W_t\) --- mutually
          incompatible active subgraphs of working memory.
    \item \textbf{WM--Belief}: \(X \subseteq W_t\),
          \(Y \subseteq B \setminus W_t\) --- active WM content conflicts
          with stored inactive belief.
    \item \textbf{Structural}: \(X \subseteq W_t\) --- \(X\) is internally
          incompatible under \(\mathcal{S}_B\).
\end{enumerate}
\end{remark}

\begin{definition}[Conflict catalog constructor]
\label{def:conflict-catalog-constructor}
Let \((\mathcal{F}_B, \mathcal{S}_B)\) be the belief-graph incompatibility
constraint, let \(W_t\), \(B\), and \(\Gamma_t\) specify the current state,
and let \(\mathrm{scope}\) be a predicate on region pairs \((X,Y)\).
The \emph{conflict catalog constructor} is
\[
\mathrm{BuildConflictCatalog}(\mathcal{F}_B, \mathcal{S}_B, \mathrm{scope}, W_t, B, \Gamma_t)
=
\bigl\{ (X,Y) : X \bowtie Y \text{ under Definition~\ref{def:cross-memory-incompatibility}}
\text{ and } \mathrm{scope}(X,Y) \bigr\}.
\]
Domain-specific catalogs such as \(\mathfrak{C}_{\mathrm{self}}\) and
task-restricted confusion catalogs are scope-instantiations of
\(\mathrm{BuildConflictCatalog}\), parametrized entirely by their
\(\mathrm{scope}\) predicate.
\end{definition}

\begin{remark}[Conflict catalog scope]
\label{rem:conflict-catalog-scope}
Typical scope predicates restrict catalogs to intra-region pairs
\((X,Y \subseteq Q)\), WM--belief pairs with \(V_X \cap R \neq \varnothing\)
for a salience set \(R\), or the unrestricted WM--belief catalog
\(\mathfrak{C}_{\mathrm{wb}}(W_t,B)\) of Proposition~\ref{prop:wm-belief-conflict-catalog}.
\end{remark}

\begin{definition}[Active conflict count]
\label{def:active-conflict-count}
Let \(\mathfrak{C}_\kappa\) be a conflict catalog produced by
Definition~\ref{def:conflict-catalog-constructor} or an equivalent finite
enumeration of ordered pairs \((X, Y)\) such that \(X \bowtie Y\) under
Definition~\ref{def:cross-memory-incompatibility}.
An entry \((X,Y)\) is \emph{actively conflicting} at time \(t\) when
\[
(V_X \cup V_Y)\cap V_{W_t} \neq \varnothing,
\]
and the incompatibility witness for \(X \bowtie Y\) involves at least one WM-active node or transient grounding currently present in \(W_t\).
The \emph{active conflict count} in \(W_t\) is
\[
\kappa(\mathfrak{C}_\kappa, W_t) = \bigl|\{(X, Y) \in \mathfrak{C}_\kappa :
(X, Y) \text{ is actively conflicting at } t\}\bigr|.
\]
Domain-specific catalogs are written with descriptive subscripts; the
WM--belief instance is \(\mathfrak{C}_{\mathrm{wb}}(W_t, B)\)
(Proposition~\ref{prop:wm-belief-conflict-catalog}).
\end{definition}

\begin{remark}[Notation]
\label{rem:conflict-catalog-notation}
Conflict catalogs use Fraktur \(\mathfrak{C}\), while the control functional is \(\mathcal{K}\).
\end{remark}

\begin{definition}[Subgraph coherence]
\label{def:subgraph-coherence}
Given a designated subgraph \(Q \subseteq B\) and a conflict catalog
\(\mathfrak{C}_Q\) of region pairs whose entries lie in
\(V_Q \times V_Q\), the \emph{coherence} of \(Q\) in \(W_t\) is
\[
\mathcal{H}(Q, \mathfrak{C}_Q, W_t) =
\begin{cases}
1 & \text{if } \mathfrak{C}_Q = \varnothing, \\
1 - \dfrac{\kappa(\mathfrak{C}_Q, W_t)}{|\mathfrak{C}_Q|} & \text{otherwise}.
\end{cases}
\]
When \(\mathfrak{C}_Q = \varnothing\), coherence is \(1\) by convention
(vacuous coherence: no catalogued conflicts to engage).
Otherwise, \(\mathcal{H} = 1\) means no active conflicts within \(Q\);
\(\mathcal{H} = 0\) means all catalogued conflicts are simultaneously
active.
\end{definition}

\begin{remark}[Encoding gate]
\label{rem:encoding-gate}
Encoding proceeds only when no WM--belief conflict is actively engaged:
\(\kappa(\mathfrak{C}_{\mathrm{wb}}(W_t, B), W_t) = 0\)
(Proposition~\ref{prop:wm-belief-conflict-catalog} and
Definition~\ref{def:active-conflict-count}); otherwise the transient node
remains in \(W_t\) until revision, decay, or re-grounding clears the
conflict.
\end{remark}

\begin{definition}[Subgraph involvement]
\label{def:subgraph-involvement}
Given a designated subgraph \(Q = (V_Q, E_Q) \subseteq B\) and a
working-memory state \(W_t\), the \emph{involvement} of \(Q\) in \(W_t\) is
\[
\mathcal{I}(Q, W_t) = S_{\mathrm{sat}}(V_Q, W_t) = \frac{|\mathrm{active}(V_Q, W_t)|}{|V_Q|}.
\]
We say \(W_t\) is \emph{\(Q\)-involving} when \(\mathcal{I}(Q, W_t) > 0\),
and \emph{fully \(Q\)-involving} when \(\mathcal{I}(Q, W_t) = 1\).
\end{definition}

Domain-specific involvement measures instantiate
\(\mathcal{I}\) (Definition~\ref{def:subgraph-involvement})
with domain-appropriate subgraphs \(Q\); these are introduced
where each domain is defined.

In working-memory graphs, \(\alpha_t^{W}\) is the primary dynamic quantity. Edge weights \(w_t\) on \(E_{W_t}\) may be derived from endpoint activations, e.g.\ \(w_t(u,v) = \alpha_t^{W}(u)\,\alpha_t^{W}(v)\) for \(e=(u,v)\in E_{W_t}\).

\begin{definition}[Signed subgraph evaluation]
\label{def:signed-subgraph-evaluation}
Given a designated subgraph \(Q \subseteq B\) and a signing function
\(\sigma_Q : V_Q \rightarrow \mathbb{R}\) that assigns positive weights
to favorable nodes and negative weights to unfavorable ones, the
\emph{signed evaluation} of \(Q\) in \(W_t\) is
\[
\mathcal{E}(Q, \sigma_Q, W_t) = \sum_{v \in \mathrm{active}(V_Q, W_t)}
\sigma_Q(v)\,\alpha_t^{W}(v).
\]
Grounding scores \(\sigma_t\) (Definition~\ref{def:wm-belief-collation}) use a separate symbol.
\end{definition}

\begin{remark}[Affective valence as evaluative interpretation]
\label{rem:affective-valence-interpretation}
Given a designated subgraph \(Q \subseteq B\)---for example an attachment subgraph \(A_{\mathrm{att}}\)---affective valence is read as a signed evaluative interpretation over involvement \(\mathcal{I}(Q,W_t)\), subgraph coherence \(\mathcal{H}(Q,\mathfrak{C}_Q,W_t)\), and signed evaluation \(\mathcal{E}(Q,\sigma_Q,W_t)\) when a signing function \(\sigma_Q\) is supplied. Domain-specific conflict catalogs \(\mathfrak{C}_Q\) specialize the generic conflict machinery (Definition~\ref{def:conflict-catalog-constructor}); no separate valence operator family is introduced beyond combining maps over these quantities.
\end{remark}

\begin{remark}[Attachment goal and satisfaction]
\label{def:attachment-goal}
An \emph{attachment goal} is a designated subgraph \(A_{\mathrm{att}}=(V_{A_{\mathrm{att}}},E_{A_{\mathrm{att}}}) \subseteq B\) whose continued accessibility or working-memory presence is preferred. Given a retrievability predicate \(\mathrm{retrievable}(v,B,W_t)\), \emph{attachment satisfaction} is the involvement-style quantity
\[
S_{\mathrm{att}}(A_{\mathrm{att}},W_t)
=
\frac{\bigl|\{v \in V_{A_{\mathrm{att}}} : \mathrm{active}(v,W_t)\ \text{or}\ \mathrm{retrievable}(v,B,W_t)\}\bigr|}{|V_{A_{\mathrm{att}}}|},
\]
a specialization of Definitions~\ref{def:node-set-satisfaction} and~\ref{def:subgraph-involvement}.
\end{remark}

\begin{definition}[Working-memory transition operator]
\label{def:wm-transition}
A \emph{working-memory transition operator} is any function
\[
U : \mathcal{W} \times \mathcal{B} \rightarrow \mathcal{W},
\]
that maps a working-memory state \(W_t\) and belief graph \(B\) to a successor state \(W_{t+1} = U(W_t, B)\), updating both graph structure and the activation maps \(\alpha_t^B\) and \(\alpha_t^{W}\). Input-driven, retrieval-driven, decay, and integration operations are instances of \(U\).
\end{definition}

\begin{definition}[Subgraph salience bias]
\label{def:subgraph-salience-bias}
Given a designated subgraph \(Q \subseteq B\) and a salience coefficient \(\beta_Q \geq 1\), a \emph{subgraph salience bias} parametrizes the working-memory transition operator \(U\) (Definition~\ref{def:wm-transition}) by supplying a salience map \(\beta^{\mathrm{sal}}_Q : V_B \rightarrow \mathbb{R}_{\ge 1}\) with \(\beta^{\mathrm{sal}}_Q(v) = \beta_Q\) for \(v \in V_Q\) and \(\beta^{\mathrm{sal}}_Q(v) = 1\) otherwise (distinct from the scalar coefficient \(\beta_Q\)).
Under this parametrization, an \emph{activation update map} \(U_{\mathrm{act}}(\alpha_t^{W},\beta^{\mathrm{sal}}_Q)\) applies multiplicative gain on activation update and reduced decay on nodes in \(V_Q\) (distinct from the WM transition operator \(U\), Definition~\ref{def:wm-transition}):
\[
\alpha_{t+1}^{W}(v) = U_{\mathrm{act}}\bigl(\alpha_t^{W}, \beta^{\mathrm{sal}}_Q\bigr)(v) =
\begin{cases}
\beta_Q \cdot (\alpha_t^{W})^{\mathrm{base}}(v) & v \in V_Q, \\
(\alpha_t^{W})^{\mathrm{base}}(v) & v \notin V_Q,
\end{cases}
\]
where \((\alpha_t^{W})^{\mathrm{base}}\) is the WM activation profile produced by the unbiased transition step, and decay on \(v \in V_Q\) is scaled by \(1/\beta_Q\) relative to nodes outside \(V_Q\) of equivalent base activation.
Setting \(\beta_Q = 1\) recovers standard activation dynamics with no subgraph bias.
Domain-specific biases are parametrizations of this definition by choice of \(Q\) and \(\beta_Q\), governed by the single multiplicative-gain and decay-scaling rule specified above.
\end{definition}

\begin{remark}[Attachment as salience parametrization]
\label{def:attachment-salience-param}
Let \(A_{\mathrm{att}}=(V_{A_{\mathrm{att}}},E_{A_{\mathrm{att}}}) \subseteq B\) be a designated attachment subgraph. Attachment is modeled as a parametrization of \(U\) with salience coefficient \(\beta_{\mathrm{att}} \ge 1\) on \(V_{A_{\mathrm{att}}}\), or a node-wise map \(\beta^{\mathrm{sal}}_{\mathrm{att}} : V_B \to \mathbb{R}_{\ge 1}\) that amplifies retrieval and reduces decay on attached nodes. Attachment does not introduce a primitive beyond the salience-bias rule above; it is persistence- and retrieval-biased regulation over designated content, optionally sustained by control outputs from \(\mathcal{K}\).
\end{remark}

\begin{definition}[Parameter bundle]
\label{def:parameter-bundle}
An architecture \emph{parameter bundle} is a tuple
\[
\Theta = (\theta_W,\, \theta_{\mathrm{gnd}},\, \theta_{\mathrm{sim}},\, \theta_{\mathrm{dist}},\, \theta_{\mathrm{inv}},\, \theta_{\mathrm{coh}},\, \{\beta_Q\}_{Q \in \mathcal{Q}_{\mathrm{des}}}),
\]
where \(\theta_W\) governs WM membership (Definition~\ref{def:activation-function}), \(\theta_{\mathrm{gnd}}\) governs grounding correspondence in WM--belief collation (Definition~\ref{def:wm-belief-collation}), \(\theta_{\mathrm{sim}}\) governs graph-similarity predicates, \(\theta_{\mathrm{dist}}\) governs distraction thresholds on \(\Phi_{\mathrm{dist}}\), \(\theta_{\mathrm{inv}}\) governs involvement predicates on \(\mathcal{I}\), \(\theta_{\mathrm{coh}}\) governs coherence predicates on \(\mathcal{H}\), and \(\{\beta_Q\}_{Q \in \mathcal{Q}_{\mathrm{des}}}\) supplies salience coefficients for designated subgraphs \(Q \subseteq B\) under Definition~\ref{def:subgraph-salience-bias} (\(\mathcal{Q}_{\mathrm{des}}\) is the family of designated subgraphs, distinct from any single subgraph symbol \(Q\)).
Throughout the paper, threshold directions are fixed by quantity type: activation and similarity predicates use lower bounds; conflict and distraction predicates use upper bounds or nonzero tests; coherence predicates use lower bounds; and binary indicators are defined explicitly case by case. Application sections may assign threshold values, but they do not alter these comparison directions.
Application sections instantiate \(\Theta\) by assigning values to these slots.
\end{definition}

\begin{remark}[Development as parameter regime]
\label{rem:development-parameter-regime}
Developmental differences are modeled as regime changes over already defined architectural degrees of freedom, including working-memory capacity bounds \((C_V,C_E)\), admissible relation-type subsets of \(\mathcal{R}_B\), allowable payload complexity and recursive depth, salience and threshold parameters in \(\Theta\), and the class of belief-update applications permitted under \(\Delta_{\bullet}\). Developmental-stage theories should therefore be read as qualitative or quantitative regimes over these parameters rather than as additions to the ontology.
\end{remark}

\begin{definition}[Fragmentation]
\label{def:fragmentation}
Let \(W_t\) have connected components \(C_1, \dots, C_k\) (ignoring edge direction). The \emph{fragmentation component} of the graph property functional is
\[
\Phi_{\mathrm{frag}}(W_t) = k.
\]
We write \(F(W_t) = \Phi_{\mathrm{frag}}(W_t)\) and say that \(W_t\) is \emph{fragmented} if \(F(W_t) > 1\).
\end{definition}

\begin{definition}[Belief-update operator]
\label{def:belief-update-operator}
A \emph{belief-update operator} is any function
\[
\Delta_{\bullet} : \mathcal{B} \times \mathcal{W}^* \rightarrow \mathcal{B},
\]
that maps a belief graph \(B\) and a finite sequence of working-memory states \((W_{t-k}, \dots, W_t)\) to an updated belief graph \(B'\), using operations such as node and edge addition, deletion, retyping, and modification of representational payloads.
Two canonical application interfaces are:
\begin{align*}
\mathrm{ApplyMicroUpdate}(B, (W_{t-k},\dots,W_t), \mathfrak{C}_\kappa, \Theta)
&\mapsto B' \\
\mathrm{ApplyMacroUpdate}(B, (W_{t-k},\dots,W_t), \mathfrak{C}_\kappa, \Theta, \tau_{\mathrm{trans}})
&\mapsto B',
\end{align*}
where \(\mathfrak{C}_\kappa\) is a conflict catalog
(Definition~\ref{def:conflict-catalog-constructor}),
\(\Theta\) is a parameter bundle (Definition~\ref{def:parameter-bundle}),
and \(\tau_{\mathrm{trans}} \in \{\mathrm{strict}, \mathrm{permit}\}\) controls whether
ill-formed intermediate graphs are allowed during multi-step restructuring (distinct from the edge-type map \(\tau\) on base graphs, Definition~\ref{def:base-graph}).
\(\mathrm{ApplyMicroUpdate}\) modifies a proper subgraph while preserving global topology and requires \(\kappa(\mathfrak{C}_{\mathrm{wb}}(W_t,B),W_t)=0\) at commit (Remark~\ref{rem:encoding-gate}).
\(\mathrm{ApplyMacroUpdate}\) may reorganize entire mental models or designated subgraphs \(Q \subseteq B\); when \(\tau_{\mathrm{trans}}=\mathrm{permit}\), transient violations of \((\mathcal{F}_B,\mathcal{S}_B)\) are allowed provided the committed output is well-formed.
Specific instances are introduced where the corresponding domain is defined.
A belief-update application is \emph{micro} if it is realizable as \(\mathrm{ApplyMicroUpdate}\); it is \emph{macro} if it requires \(\mathrm{ApplyMacroUpdate}\) with nontrivial restructuring.
\end{definition}

\begin{remark}[Procedurally compiled subgraph]
\label{def:procedurally-compiled-subgraph}
A procedurally compiled subgraph is a designated belief subgraph \(S^{\mathrm{cmp}}=(V_{S^{\mathrm{cmp}}},E_{S^{\mathrm{cmp}}}) \subseteq B\) produced by repeated use, chunking, or schema consolidation, such that it can be retrieved into working memory with reduced explicit intermediate expansion. In skill-acquisition terms, this is a procedural specialization of recursive compression and chunk formation already permitted under \(\Delta_{\bullet}\).
\end{remark}

\begin{definition}[Trajectory functional]
\label{def:trajectory-functional}
A \emph{trajectory functional} is any mapping
\[
\Psi : \mathcal{W}^* \rightarrow \mathcal{Z},
\]
from finite sequences of working-memory graphs to a representation space \(\mathcal{Z}\) (e.g., control signals, temporally annotated graphs, or scalar summaries).
Concrete trajectory functionals are introduced in subsequent subsections.
\end{definition}

\begin{definition}[Trajectory label sets]
\label{def:trajectory-label-sets}
The canonical temporal and provenance label sets are
\[
\mathfrak{L}_{\mathrm{temp}}=\{\mathrm{static},\mathrm{dynamic},\mathrm{persisting},\mathrm{changing}\},
\qquad
\mathfrak{L}_{\mathrm{prov}}=\{\mathrm{input},\mathrm{retrieval},\mathrm{integration},\mathrm{decay},\mathrm{chunk}\}.
\]
A temporal trajectory functional is a mapping
\[
\mathcal{T} : \mathcal{W}^* \times 2^{V_B} \rightarrow \mathfrak{L}_{\mathrm{temp}}
\]
that assigns each designated subgraph \(X\) a label in \(\mathfrak{L}_{\mathrm{temp}}\) relative to a finite trajectory \((W_{t-k},\dots,W_t)\), satisfying time-consistency: labels depend only on the restriction of the trajectory to nodes in \(X\) and must be invariant under padding with identical copies of \(W_t\).
A provenance trajectory functional is a mapping
\[
\Pi : \mathcal{W}^* \times 2^{V_B} \rightarrow \mathfrak{L}_{\mathrm{prov}}
\]
that classifies the most recent transition-induced origin of subgraph content under the working-memory transition operator \(U\) (Definition~\ref{def:wm-transition}).
Both \(\mathcal{T}\) and \(\Pi\) are instances of the generic trajectory functional \(\Psi\) (Definition~\ref{def:trajectory-functional}) with representation space \(\mathcal{Z} = \mathfrak{L}_{\mathrm{temp}} \times 2^{V_B}\) or \(\mathfrak{L}_{\mathrm{prov}} \times 2^{V_B}\).
Architecture choices supply the decision rules that assign labels; the label sets and signatures are fixed by this definition.
\end{definition}

\noindent\textbf{Remark.}
Working-memory dynamics form a sequence \(\{W_t\}_{t \in T_{\mathrm{step}}}\) over discrete time steps \(T_{\mathrm{step}} = \{0, 1, 2, \dots\}\) (distinct from task specification index \(T\)), together with activation maps \(\alpha_t^B\) and \(\alpha_t^{W}\) and transition operator \(U\). Each step \(W_{t+1} = U(W_t, B)\) results from four operation classes: input-driven activation, retrieval-driven activation, decay and pruning, and integration. Perceptual input nodes (Definition~\ref{def:perceptual-input-node}) and encoding (Remark~\ref{rem:encoding-gate}) formalize the input-to-belief stages of this pipeline.

\begin{remark}[Belief-update admissibility]
\label{rem:belief-update-admissibility}
The encoding gate of Remark~\ref{rem:encoding-gate} is enforced on
\emph{commit} steps of \(\mathrm{ApplyMicroUpdate}\)
(Definition~\ref{def:belief-update-operator}).
\(\mathrm{ApplyMacroUpdate}\) with \(\tau_{\mathrm{trans}}=\mathrm{permit}\)
may traverse transient ill-formed intermediate graphs provided the
committed output restores well-formedness under \((\mathcal{F}_B,\mathcal{S}_B)\).
\end{remark}

\subsection{Self-image and self-related processing}
\label{sec:self}

The self-related constructs below are direct specializations of the generic toolkit introduced in Section~\ref{sec:generic-functionals}: self-salience is a node-set satisfaction measure, self-involvement is a subgraph involvement measure, self-evaluation is a signed subgraph evaluation, self-consistency is a subgraph coherence measure, and self-update is a belief-update operation on the designated self-image subgraph \(S_{\mathrm{self}}\).

\begin{definition}[Self-image subgraph]
\label{def:self-image-subgraph}
Let \(B = (V_B, E_B, \mathcal{R}_B, \tau_B, w_B)\) be the belief graph. A \emph{self-image subgraph} is a connected subgraph
\[
S_{\mathrm{self}} = (V_{\mathrm{self}}, E_{\mathrm{self}}) \subseteq B
\]
distinguished by an anchor node \(v_{\mathrm{self}} \in V_{\mathrm{self}}\) whose payload denotes the agent as self.

The node set \(V_{\mathrm{self}}\) contains nodes representing:
\begin{enumerate}
    \item bodily attributes and interoceptive states,
    \item autobiographical episodes and remembered events,
    \item trait-like self-descriptors,
    \item goals, intentions, commitments, and values,
    \item social roles and identity-linked relational nodes,
    \item meta-representations of the self, including beliefs about how the self is seen by others.
\end{enumerate}

The edge set \(E_{\mathrm{self}}\) may include causal, containment, temporal, associative, evidential, and spatial relations, restricted to the relation classes already admitted by Definition~\ref{def:edge-type-taxonomy}. The self-image subgraph is well-formed when it is connected, internally compatible under \((\mathcal{F}_B,\mathcal{S}_B)\), and anchored by either one designated self node or one designated self anchor cluster treated as a single identity unit.
\end{definition}

\begin{remark}[Self-image interpretation]
\label{rem:self-image-interpretation}
The self-image subgraph is not a special kind of graph outside the ontology; it is a designated subgraph of \(B\) selected by functional role rather than by a new formal primitive. Its content may change over time through encoding, consolidation, revision, suppression, or reweighting, but its identity is preserved by the anchor structure and by the continuity of self-referential links.
\end{remark}

\begin{definition}[Self-related salience]
\label{def:self-salience}
A subgraph \(Q \subseteq B\) is \emph{self-related} if \(V_Q \cap V_{\mathrm{self}} \neq \varnothing\) or if \(Q\) is linked to \(S_{\mathrm{self}}\) by one or more self-referential edges in \(E_B\).
The self-salience of \(Q\) at time \(t\) is an activation-weighted measure of how much self-related content in \(Q\) is currently present in working memory:
\[
S_{\mathrm{self-sat}}(Q, W_t) = S_{\mathrm{sat}}(V_Q \cap V_{\mathrm{self}}, W_t)
\]
when \(V_Q \cap V_{\mathrm{self}} \neq \varnothing\) (Definition~\ref{def:node-set-satisfaction}), and \(0\) otherwise.
The score is activation-weighted through the underlying node-set satisfaction measure.
\end{definition}

\begin{definition}[Self-involvement]
\label{def:self-involvement}
The \emph{self-involvement} of the current working-memory state \(W_t\) is
\[
\mathcal{I}_{\mathrm{self}}(W_t) = \mathcal{I}(S_{\mathrm{self}}, W_t)
= \frac{|\mathrm{active}(V_{\mathrm{self}}, W_t)|}{|V_{\mathrm{self}}|}
\]
(Definition~\ref{def:subgraph-involvement}).
We say that \(W_t\) is \emph{self-involving} when \(\mathcal{I}_{\mathrm{self}}(W_t) > 0\), and \emph{fully self-involving} when \(\mathcal{I}_{\mathrm{self}}(W_t) = 1\).
\end{definition}

\begin{definition}[Self-evaluation]
\label{def:self-evaluation}
Let \(\sigma_{\mathrm{eval}} : V_{\mathrm{self}} \rightarrow \mathbb{R}\) be a valuation assigning positive weight to desirable self-nodes and negative weight to undesirable self-nodes.
The \emph{self-evaluation} of \(W_t\) is
\[
\mathcal{E}_{\mathrm{self}}(W_t) =
\mathcal{E}(S_{\mathrm{self}}, \sigma_{\mathrm{eval}}, W_t) =
\sum_{v \in \mathrm{active}(V_{\mathrm{self}}, W_t)} \sigma_{\mathrm{eval}}(v)\,\alpha_t^{W}(v).
\]
(Definition~\ref{def:signed-subgraph-evaluation}.)
This quantity summarizes the affectively weighted activation of the self-image.
\end{definition}

\begin{definition}[Self-consistency]
\label{def:self-consistency}
Let
\[
\mathfrak{C}_{\mathrm{self}} =
\mathrm{BuildConflictCatalog}(\mathcal{F}_B, \mathcal{S}_B, \mathrm{scope}_{\mathrm{self}}, W_t, B, \Gamma_t),
\]
where \(\mathrm{scope}_{\mathrm{self}}\) retains region pairs internal to \(S_{\mathrm{self}}\) (Definition~\ref{def:conflict-catalog-constructor}).
The \emph{self-consistency} of \(S_{\mathrm{self}}\) in \(W_t\) is
\[
\mathcal{H}_{\mathrm{self}}(W_t)
=
\mathcal{H}(S_{\mathrm{self}}, \mathfrak{C}_{\mathrm{self}}, W_t).
\]
(Definition~\ref{def:subgraph-coherence}.)
A value of \(1\) means that no catalogued self-conflicts are actively engaged, while \(0\) means that all catalogued self-conflicts are active.
\end{definition}

\begin{remark}[Self-conflict source]
\label{rem:self-conflict-source}
Self-conflict may arise from incompatible self-descriptors, contradictory autobiographical claims, misaligned goals, or incompatible role commitments.
The self-consistency measure evaluates conflicts identified by the general incompatibility framework (Definition~\ref{def:incompatibility}) restricted to \(S_{\mathrm{self}}\).
\end{remark}

\begin{definition}[Self-update]
\label{def:self-update}
A \emph{self-update} is an instance of \(\mathrm{ApplyMicroUpdate}\) or \(\mathrm{ApplyMacroUpdate}\) (Definition~\ref{def:belief-update-operator}) whose application revises the self-image subgraph \(S_{\mathrm{self}} \subseteq B\) in response to self-relevant input.
The update may:
\begin{enumerate}
    \item add a new self-node,
    \item delete an obsolete self-node,
    \item strengthen or weaken a self-edge,
    \item retype a self-relation,
    \item resolve a self-conflict,
    \item compress a repeated pattern into a higher-level self-schema.
\end{enumerate}
We write \(\Delta_{\mathrm{self}}\) for a self-update instance.
\end{definition}

\begin{remark}[Self-update admissibility]
Micro-scale trait or edge revisions use \(\mathrm{ApplyMicroUpdate}\) and respect the encoding gate at commit. Macro-scale self-schema reorganization may use \(\mathrm{ApplyMacroUpdate}\) with \(\tau_{\mathrm{trans}}=\mathrm{permit}\) when transient ill-formed intermediate states are required. In all cases, the committed output must preserve graph well-formedness under \((\mathcal{F}_B,\mathcal{S}_B)\).
\end{remark}

\begin{definition}[Self-referential activation bias]
\label{def:self-referential-bias}
A \emph{self-referential activation bias} instantiates subgraph salience bias (Definition~\ref{def:subgraph-salience-bias}) for \(Q = S_{\mathrm{self}}\) with coefficient \(\beta_{\mathrm{self}}\) from the parameter bundle \(\Theta\) (Definition~\ref{def:parameter-bundle}), inducing salience map \(\beta^{\mathrm{sal}}_{S_{\mathrm{self}}}\):
\[
\alpha_{t+1}^{W}(v) = U_{\mathrm{act}}\bigl(\alpha_t^{W}, \beta^{\mathrm{sal}}_{S_{\mathrm{self}}}\bigr)(v),
\]
so that nodes \(v \in V_{\mathrm{self}}\) receive multiplicative gain and reduced decay under \(U\) as specified in Definition~\ref{def:subgraph-salience-bias}.
This bias increases the likelihood that self-related nodes enter working memory and remain active across short time intervals.
\end{definition}

\begin{remark}[Self-related processing loop]
\label{rem:self-processing-loop}
Self-related processing proceeds as a loop between activation, selection, evaluation, and update.
Self-related content enters working memory through salience bias, is evaluated by the current self-image, and may trigger update when inconsistency persists or when new evidence systematically improves a rival self-organization \cite{Moutoussis2014}.
\end{remark}

\begin{definition}[Narrative self-coherence]
\label{def:narrative-self-coherence}
For a finite working-memory trajectory \((W_{t_0}, W_{t_1}, \dots, W_{t_n}) \in \mathcal{W}^*\), narrative self-coherence instantiates involvement and coherence functionals (Definitions~\ref{def:subgraph-involvement} and~\ref{def:subgraph-coherence}) over the self-image subgraph:
\[
N_{\mathrm{self}}(W_{t_0}, \dots, W_{t_n}) =
\frac{1}{n+1} \sum_{i=0}^{n} \mathcal{I}_{\mathrm{self}}(W_{t_i}) \cdot \mathcal{H}_{\mathrm{self}}(W_{t_i}),
\]
where \(\mathcal{I}_{\mathrm{self}}(W_{t_i}) = \mathcal{I}(S_{\mathrm{self}}, W_{t_i})\) (Definition~\ref{def:self-involvement}).
This score is high when self-related content is both frequently active and internally consistent across time.
\end{definition}

\begin{remark}[Narrative interpretation]
\label{rem:narrative-interpretation}
Narrative self-coherence is a composite diagnostic rather than a primitive construct.
It combines the extent to which self-content is active with the degree to which that content remains conflict-minimized over time.
\end{remark}

\subsection{Awareness functional}
\label{sec:awareness}

Awareness is the capacity of the system to have access to any part of \(W_t\)---fragmented or not, partial or whole. The awareness functional \(\mathcal{A}\) maps each working-memory state to the subgraph currently accessed; the control functional \(\mathcal{K}\), defined below, inspects working-memory trajectories and emits control operations, operating on whatever portion of \(W_t\) is in awareness.

\begin{definition}[Awareness functional]
\label{def:awareness-functional}
An \emph{awareness functional} is a mapping
\[
\mathcal{A} : \mathcal{W} \rightarrow 2^{\mathcal{W}},
\]
where \(2^{\mathcal{W}}\) denotes the power set of subgraphs of working-memory graphs. For a working-memory state \(W_t\), the value \(\mathcal{A}(W_t)\) is a subgraph \(A_t \subseteq W_t\)---the portion of working memory that the system currently has access to. A node or subgraph \(X \subseteq W_t\) is \emph{in awareness} at \(t\) if \(X \subseteq \mathcal{A}(W_t)\).

The system is said to be \emph{aware} at \(t\) if \(\mathcal{A}(W_t) \neq \varnothing\), i.e., at least some content of \(W_t\) is accessed. No connectivity constraint is imposed: \(\mathcal{A}(W_t)\) may be a single node, a fragment, or any subgraph of \(W_t\).

Unlike trajectory functionals \(\Psi\) (Definition~\ref{def:trajectory-functional}), which map finite WM sequences to summary representations, \(\mathcal{A}\) is a state-level access selector: its input is a single \(W_t\), not a trajectory \((W_{t-k},\dots,W_t)\). Temporal and regulatory summaries over trajectories---including control via \(\mathcal{K}\) (Section~\ref{sec:control}) and transformed-awareness annotations via \(\mathcal{T}\) and \(\Pi\) (Section~\ref{sec:transformed-awareness})---are supplied by the \(\Psi\)-family; \(\mathcal{A}\) specifies which subgraph of \(W_t\) those summaries annotate at time \(t\).
\end{definition}

Awareness thus provides the basic access relation to \(W_t\); higher states such as transformed awareness and insightful perception are defined in the following subsection.

\begin{definition}[Total awareness]
\label{def:total-awareness}
A working-memory state \(W_t\) is in \emph{total awareness} at \(t\) if
\[
\mathcal{A}(W_t) = W_t,
\]
i.e., the awareness functional covers the entirety of the current working-memory graph, including all nodes and edges in the current working-memory graph. Total awareness is the special case of awareness in which no content of \(W_t\) is excluded from access.
\end{definition}

\subsection{Transformed awareness}
\label{sec:transformed-awareness}

The constructs in this subsection form a hierarchy of binary predicates over awareness and trajectory annotations, built from the awareness functional \(\mathcal{A}\), trajectory functionals \(\mathcal{T}\) and \(\Pi\), and the conflict count \(\kappa\).
\emph{Transformed awareness} is awareness \(\mathcal{A}(W_t)\) together with a temporal annotation on an accessed subgraph.
\emph{Temporal objectification} holds when that annotation is temporal; \emph{psychological time} holds when at least one accessed subgraph is temporally objectified.
\emph{Provenance objectification} and \emph{conflict objectification} are parallel predicates defined by provenance and conflict annotations respectively.
\emph{Insightful perception} is the disjunction of the objectification predicates (together with evaluative annotation); \emph{total insight} is insightful perception under total awareness.
These predicates are existentially quantified over finite recent trajectories but are evaluated as binary outcomes at the current time \(t\).

\begin{definition}[Transformed awareness]
\label{def:transformed-awareness}
Let \(W_t\) be a working-memory graph at time \(t\), and let \(\mathcal{A}(W_t) \subseteq W_t\) be the subgraph in awareness. The state \(W_t\) exhibits \emph{transformed awareness} if there exists \(k \geq 0\), a finite trajectory \((W_{t-k},\dots,W_t)\) on the interval \([t-k,t]\), a subgraph \(X \subseteq \mathcal{A}(W_t)\), and a temporal trajectory functional \(\mathcal{T}\) (Definition~\ref{def:trajectory-label-sets}) such that \(\mathcal{T}(W_{t-k},\dots,W_t, X)\) assigns to \(X\) a label in \(\mathfrak{L}_{\mathrm{temp}}\).
\end{definition}

\begin{definition}[Temporal objectification]
\label{def:temporal-objectification}
Let \(W_t\) be a working-memory state, let \(k \geq 0\), and let \((W_{t-k},\dots,W_t)\) be a finite trajectory on \([t-k,t]\). A subgraph \(X \subseteq W_t\) is \emph{temporally objectified} at time \(t\) if
\[
X \subseteq \mathcal{A}(W_t)
\]
and a temporal trajectory functional \(\mathcal{T}\) (Definition~\ref{def:trajectory-label-sets}) assigns to \(X\) a label in \(\mathfrak{L}_{\mathrm{temp}}\) relative to \((W_{t-k},\dots,W_t)\).
Thus, temporal objectification is a specialized form of transformed awareness.
\end{definition}

\begin{definition}[Psychological time]
\label{def:psychological-time}
Let \(W_t\) be a working-memory state and let \(k \geq 0\). A state \(W_t\) exhibits \emph{psychological time} on \([t-k,t]\) if there exists a subgraph \(X \subseteq \mathcal{A}(W_t)\) and a temporal trajectory functional \(\mathcal{T}\) (Definition~\ref{def:trajectory-label-sets}) such that \(\mathcal{T}(W_{t-k},\dots,W_t, X)\) assigns to \(X\) a label in \(\mathfrak{L}_{\mathrm{temp}}\).
Equivalently, psychological time is present exactly when at least one subgraph in awareness is temporally objectified.
\end{definition}

\begin{definition}[Provenance objectification]
\label{def:provenance-objectification}
Let \(W_t\) be a working-memory state, let \(k \geq 0\), and let \((W_{t-k},\dots,W_t)\) be a finite trajectory on \([t-k,t]\). A subgraph \(X \subseteq W_t\) is \emph{provenance-objectified} at time \(t\) if
\[
X \subseteq \mathcal{A}(W_t)
\]
and a provenance trajectory functional \(\Pi\) (Definition~\ref{def:trajectory-label-sets}) assigns to \(X\) a label in \(\mathfrak{L}_{\mathrm{prov}}\) relative to \((W_{t-k},\dots,W_t)\).
Thus, provenance objectification is a specialized form of transformed awareness.
\end{definition}

\begin{definition}[Conflict objectification]
\label{def:conflict-objectification}
Let \(W_t\) be a working-memory state and let \(\mathfrak{C}_\kappa\) be a domain-specific conflict catalog (Definition~\ref{def:active-conflict-count}) whose entries are defined under Definition~\ref{def:cross-memory-incompatibility}. A subgraph \(X \subseteq W_t\) is \emph{conflict-objectified} at time \(t\) if \(X \subseteq \mathcal{A}(W_t)\) and \(X\) is counted as actively conflicting by \(\kappa(\mathfrak{C}_\kappa,W_t)\).
\end{definition}

\begin{definition}[Conflict-aware transformed awareness]
\label{def:conflict-aware-transformed-awareness}
The state \(W_t\) exhibits \emph{conflict-aware transformed awareness} if there exists a subgraph \(X \subseteq W_t\) that is conflict-objectified at \(t\) under Definition~\ref{def:conflict-objectification}.
\end{definition}

\begin{definition}[Insightful perception]
\label{def:insightful-perception}
Let \(W_t\) be a working-memory state at time \(t\), let \(k \geq 0\), and let \((W_{t-k},\dots,W_t)\) be a finite trajectory on \([t-k,t]\). The state \(W_t\) exhibits \emph{insightful perception} at \(t\) if there exists a WM-active subgraph \(X \subseteq W_t\) (Definition~\ref{def:activation-function}) such that \(X \subseteq \mathcal{A}(W_t)\) and at least one of the following predicates holds:
\begin{enumerate}
    \item \(X\) is temporally objectified at \(t\) under Definition~\ref{def:temporal-objectification},
    \item \(X\) is conflict-objectified at \(t\) under Definition~\ref{def:conflict-objectification},
    \item there exist a designated subgraph \(Q \subseteq B\) and an architecturally designated signing function \(\sigma_Q : V_Q \to \mathbb{R}\) (e.g., \(\sigma_{\mathrm{eval}}\) when \(Q = S_{\mathrm{self}}\), Definition~\ref{def:self-evaluation}) such that \(X \subseteq Q\) and \(\mathcal{E}(Q,\sigma_Q,W_t) \neq 0\) with nonzero contribution from \(\mathrm{active}(V_X,W_t)\) (Definition~\ref{def:signed-subgraph-evaluation}),
    \item \(X\) is provenance-objectified at \(t\) under Definition~\ref{def:provenance-objectification}.
\end{enumerate}
\end{definition}

\begin{definition}[Total insight]
\label{def:total-insight}
Let \(W_t\) be a working-memory state at time \(t\), let \(k \geq 0\), and let \((W_{t-k},\dots,W_t)\) be a finite trajectory on \([t-k,t]\). The state \(W_t\) exhibits \emph{total insight} at \(t\) if
\[
\mathcal{A}(W_t)=W_t
\]
and at least one predicate in Definition~\ref{def:insightful-perception} holds with \(X=W_t\).
Equivalently, total insight is insightful perception on the whole working-memory graph under total awareness (Definition~\ref{def:total-awareness}).
\end{definition}

\begin{remark}
Definitions~\ref{def:transformed-awareness}--\ref{def:total-insight} annotate accessed subgraphs with temporal, provenance, conflict, and evaluative descriptors drawn from the formal architecture.
\end{remark}

\begin{remark}[Annotated transformed awareness]
\label{rem:annotated-transformed-awareness}
The higher-order annotations available within transformed awareness may include:
\begin{enumerate}
    \item temporal labels supplied by \(\mathcal{T}\) on \(\mathfrak{L}_{\mathrm{temp}}\),
    \item conflict status supplied by conflict catalogs and \(\kappa(\mathfrak{C}_\kappa,W_t)\),
    \item signed evaluations \(\mathcal{E}(Q,\sigma_Q,W_t)\) (Definition~\ref{def:signed-subgraph-evaluation}),
    \item provenance labels supplied by \(\Pi\) on \(\mathfrak{L}_{\mathrm{prov}}\) (Definition~\ref{def:trajectory-label-sets}),
    \item representational-format markers inherited from linguistic, perceptual, or graph payloads (Definition~\ref{def:representational-payload}).
\end{enumerate}
These annotations are descriptors defined over accessed subgraphs.
\end{remark}

\subsection{Control functional}
\label{sec:control}

The control functional is the regulation-level counterpart of the awareness functional. Where \(\mathcal{A}\) selects the subgraph currently accessed at time \(t\), the control functional maps recent working-memory trajectories to operations that alter subsequent processing.

\begin{definition}[Control functional]
\label{def:control-functional}
A \emph{control functional} is a trajectory functional
\[
\mathcal{K} : \mathcal{W}^* \rightarrow \mathfrak{O},
\]
where \(\mathfrak{O}\) is a space of control operations on working memory and retrieval dynamics.
\end{definition}

\begin{definition}[Control operation]
\label{def:control-operation}
A \emph{control operation} is any operation emitted by \(\mathcal{K}\) that modulates subsequent working-memory evolution under the transition operator \(U\) (Definition~\ref{def:wm-transition}). Typical instances include:
\begin{enumerate}
    \item attentional focusing on a designated subgraph \(Q\),
    \item retrieval amplification from belief graph regions linked to \(Q\),
    \item suppression or pruning of distractor subgraphs,
    \item persistence bias on designated active nodes,
    \item conflict-resolution routing when \(\kappa(\mathfrak{C}_\kappa,W_t)>0\) for a catalog \(\mathfrak{C}_\kappa\) from Definition~\ref{def:conflict-catalog-constructor},
    \item self-directed regulation when self-related quantities exceed thresholds in the parameter bundle \(\Theta\) (Definition~\ref{def:parameter-bundle}).
\end{enumerate}
The set of all admissible control operations is denoted \(\mathfrak{O}\), the codomain of \(\mathcal{K}\). A control operation \(c \in \mathfrak{O}\) is \emph{local} if it modifies a proper subgraph of \(W_t\) or a single mental model in \(B\); it is \emph{global} if it reorganizes the large-scale structure of \(W_t\) or \(B\).
\end{definition}

\begin{definition}[Control-sensitive trajectory]
\label{def:control-sensitive-trajectory}
A finite trajectory \((W_{t-k},\dots,W_t)\) is \emph{control-sensitive} if \(\mathcal{K}(W_{t-k},\dots,W_t)\) depends on at least one graph-property quantity defined over the trajectory, including but not limited to fragmentation \(F(W_t)\) (Definition~\ref{def:fragmentation}), self-involvement \(\mathcal{I}(S_{\mathrm{self}},W_t)\) (Definition~\ref{def:self-involvement}), self-consistency \(\mathcal{H}(S_{\mathrm{self}},\mathfrak{C}_{\mathrm{self}},W_t)\) (Definition~\ref{def:self-consistency}), signed evaluation \(\mathcal{E}(S_{\mathrm{self}},\sigma_{\mathrm{eval}},W_t)\) (Definition~\ref{def:self-evaluation}), or active conflict counts \(\kappa(\mathfrak{C}_\kappa,W_t)\) (Definition~\ref{def:active-conflict-count}).
\end{definition}

\begin{definition}[Awareness-conditioned control]
\label{def:awareness-conditioned-control}
A control functional \(\mathcal{K}\) is \emph{awareness-conditioned} if its output at time \(t\) depends only on subgraphs contained in \(\mathcal{A}(W_t)\) together with trajectory-derived annotations available within transformed awareness.
\end{definition}

\begin{proposition}[Total awareness as maximal control access]
\label{prop:total-awareness-control}
If \(\mathcal{A}(W_t)=W_t\), then any awareness-conditioned control functional has access to the full working-memory graph at time \(t\). Hence total awareness is the maximal-input case for awareness-conditioned control.
\end{proposition}

\begin{remark}
When \(\mathcal{A}(W_t)\subsetneq W_t\), awareness-conditioned control operates on a proper subgraph of the available working-memory content. This provides a formal account of regulation under partial access: control may fail not because relevant structure is absent from \(W_t\), but because it lies outside \(\mathcal{A}(W_t)\).
\end{remark}

\begin{remark}
Insightful perception and conflict-aware transformed awareness are diagnostically relevant to control because they enrich the annotations available on accessed subgraphs. In particular, temporally objectified and conflict-annotated subgraphs provide \(\mathcal{K}\) with trajectory-level information that is not available from static access alone.
\end{remark}

\begin{remark}[Toolkit instantiation roadmap]
\label{rem:toolkit-roadmap}
The generic toolkit (Section~\ref{sec:generic-functionals}) is instantiated progressively across the formal architecture: self-related processing in Section~\ref{sec:self}, awareness in Section~\ref{sec:awareness}, transformed awareness in Section~\ref{sec:transformed-awareness}, and control in Section~\ref{sec:control}. Task phenomena and ontological extensions are developed as further instantiations in subsequent sections.
\end{remark}

\section{Graph-Theoretic Characterization of Phenomena}
\label{sec:phenomena}

In this section, each phenomenon is introduced as either a scalar measure, a binary predicate, a trajectory-level diagnostic, or a mapping-level interpretation of the core toolkit from Section~\ref{sec:generic-functionals}.
Each construct introduced below has one primary mathematical status and should not switch category within the same definition.
Each construct reuses the activation maps \(\alpha_t^B\) and \(\alpha_t^{W}\), graph-property functionals \(\Phi\), incompatibility and conflict catalogs \(\mathfrak{C}_\kappa\), subgraph coherence, subgraph involvement, signed subgraph evaluation, trajectory functionals \(\Psi\), the working-memory transition operator \(U\), and belief-update operators \(\Delta_{\bullet}\).
Fragmentation, recognition, distraction, and confusion are direct \(\Phi\)-instantiations and predicates over working-memory graphs.
Cognitive load is a mapping-level translation of Sweller's Cognitive Load Theory \cite{Sweller1988,SwellerChandler1994,Sweller2010}.
Conceptual change is a belief-update process that restructures theory-bearing subgraphs of \(B\), potentially changing node labels, edge types, and compatibility relations.
Mental-state profiles below are trajectory-level diagnostics stated as propositions over configurations already defined.
The aim is to show how familiar cognitive phenomena arise as structured properties of working-memory graphs and their interaction with the belief graph \(B\).

\subsection{Fragmentation}

Throughout this section we write \(H_{\mathrm{coh}}\) for subgraph coherence \(\mathcal{H}\), \(I_{\mathrm{inv}}\) for subgraph involvement \(\mathcal{I}\), \(E_{\mathrm{eval}}\) for signed subgraph evaluation \(\mathcal{E}\), \(H_{\mathrm{ent}}\) for Shannon entropy over component masses (distinct from \(H_{\mathrm{coh}}\)), and \(L_{\mathrm{int}}, L_{\mathrm{ext}}, L_{\mathrm{ger}}\) for the cognitive-load measures in Section~\ref{sec:cognitive-load}.
Fragmentation \(F(W_t)\) is a scalar measure already defined in Definition~\ref{def:fragmentation}; the subsections below instantiate \(\Phi\) and conflict-catalog predicates over working-memory graphs.

\nestfig[Fragmented working memory with multiple weakly related components \(C_1, C_2, C_3\).]{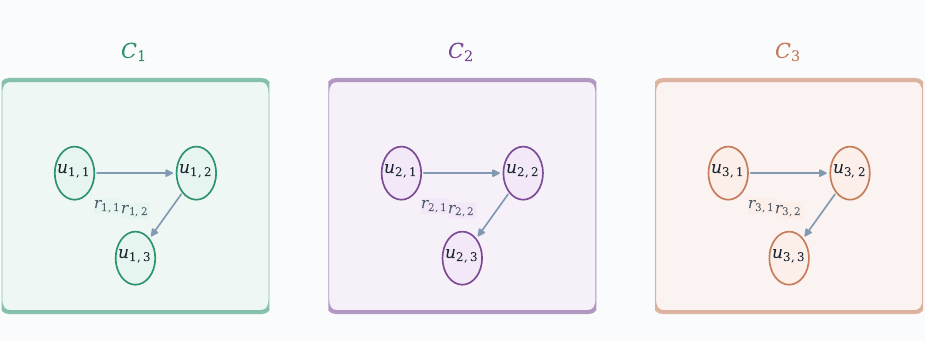}

Figure~\ref{fig:fig05_fragmentation} depicts a fragmented state with \(F(W_t) > 1\) (Definition~\ref{def:fragmentation}).

\subsection{Recognition as graph matching}

Recognition is a binary predicate over partial graph matching.
We write \(\mathrm{sim}_{\mathrm{match}}\) for match similarity, \(f_{\mathrm{match}}\) for the injective partial-match map, \(\theta_{\mathrm{sim}}\) for the recognition threshold, and \(\mathcal{M}_{\mathrm{lib}}\) for the candidate mental-model library (distinct from mental model \(M\) and from the model-library symbol \(\mathcal{M}\) when used elsewhere).

\begin{definition}[Graph similarity]
\label{def:graph-similarity}
Let \(G_{\mathrm{q}} = (V_{\mathrm{q}}, E_{\mathrm{q}})\) be a query graph and \(M = (V_M, E_M)\) a candidate graph with typed edges. A \emph{partial match} from \(G_{\mathrm{q}}\) to \(M\) is an injective mapping \(f_{\mathrm{match}} : V_{\mathrm{q}} \rightarrow V_M\) preserving a specified subset of edge types from the edge-type taxonomy (Definition~\ref{def:edge-type-taxonomy}). Let \(m_V = |f_{\mathrm{match}}(V_{\mathrm{q}})|\) and let \(m_E\) be the number of edges in \(E_{\mathrm{q}}\) whose images are present in \(E_M\) with matching types under \(f_{\mathrm{match}}\). The \emph{similarity measure} of \(G_{\mathrm{q}}\) to \(M\) under \(f_{\mathrm{match}}\) is
\[
\mathrm{sim}_{\mathrm{match}}(G_{\mathrm{q}}, M; f_{\mathrm{match}}) = \lambda_{\mathrm{match}} \frac{m_V}{|V_{\mathrm{q}}|} + (1-\lambda_{\mathrm{match}})\frac{m_E}{|E_{\mathrm{q}}|},
\]
for some \(\lambda_{\mathrm{match}} \in [0,1]\) in the parameter bundle \(\Theta\) (Definition~\ref{def:parameter-bundle}).
\end{definition}

\begin{definition}[Recognition event]
\label{def:recognition-event}
Recognition is a derived binary predicate over graph similarity (Definition~\ref{def:graph-similarity}) with threshold \(\theta_{\mathrm{sim}}\) from \(\Theta\).
Let \(W_t\) be the current working-memory graph and \(\mathcal{M}_{\mathrm{lib}}\) a library of candidate mental models in \(B\). A \emph{recognition event} occurs when there exists \(M \in \mathcal{M}_{\mathrm{lib}}\) and a partial match \(f_{\mathrm{match}}\) such that
\[
\mathrm{sim}_{\mathrm{match}}(W_t, M; f_{\mathrm{match}}) \geq \theta_{\mathrm{sim}}.
\]
A complementary \emph{failed recognition} case holds when \(X \subseteq W_t\) partially matches \(M \in \mathcal{M}_{\mathrm{lib}}\) under \(f_{\mathrm{match}}\) with \(\mathrm{sim}_{\mathrm{match}}(X, M; f_{\mathrm{match}}) \geq \theta_{\mathrm{sim}}\) but \(X \bowtie M\) in \(G^{\mathrm{col}}_t = \mathrm{Coll}(W_t, B, \Gamma_t)\) (Definitions~\ref{def:wm-belief-collation} and~\ref{def:cross-memory-incompatibility}): the episode matches structurally but is rejected by WM--belief incompatibility and contributes to \(\kappa(\mathfrak{C}_{\mathrm{wb}}(W_t, B), W_t)\) (Proposition~\ref{prop:wm-belief-conflict-catalog}).
\end{definition}

\nestfig[Recognition as partial graph matching between \(W_t\) and a stored mental model \(M\).]{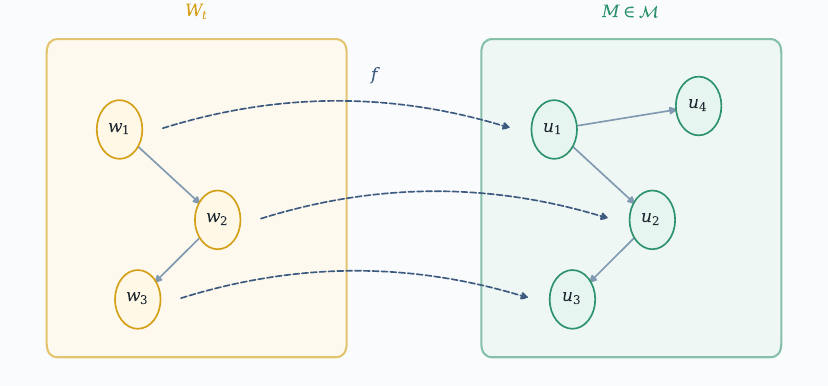}

Figure~\ref{fig:fig06_recognition} illustrates recognition as thresholded partial matching between \(W_t\) and a stored mental model \(M\).

\subsection{Distraction}

\begin{definition}[Distraction measure]
\label{def:distraction-measure}
Distraction is a derived scalar measure over \(\Phi\) (Definition~\ref{def:graph-property-functional}).
Let \(W_t\) have connected components \(C_1, \dots, C_k\) with \(k = F(W_t)\) (Definition~\ref{def:fragmentation}), and let \(p_i = \alpha_t^{W}(C_i) / \alpha_t^{W}(V_{W_t})\) be the normalized activation mass on component \(C_i\). The \emph{distraction measure} is
\[
\Phi_{\mathrm{dist}}(W_t) = H_{\mathrm{ent}}(\boldsymbol{p}) = -\sum_{i=1}^k p_i \log p_i,
\]
where \(H_{\mathrm{ent}}\) denotes Shannon entropy over the component mass vector \(\boldsymbol{p}=(p_1,\dots,p_k)\).
\end{definition}

\begin{definition}[Distracted predicate]
\label{def:distracted-predicate}
The \emph{distracted} predicate holds when \(F(W_t) > 1\) and \(\Phi_{\mathrm{dist}}(W_t) > \theta_{\mathrm{dist}}\) for the distraction threshold \(\theta_{\mathrm{dist}}\) in the parameter bundle \(\Theta\) (Definition~\ref{def:parameter-bundle}). Distraction is distinct from partial awareness: when \(W_t\) is fragmented, \(\mathcal{A}(W_t)\) may cover one component, several components, or any other subgraph of \(W_t\) without requiring a single connected access region.
\end{definition}

\subsection{Confusion}

Confusion is a binary predicate over three failure modes: disconnected task structure, active incompatibility inside the task-relevant subgraph, and WM--belief conflict.
Distraction (Definition~\ref{def:distraction-measure}) is a separate scalar measure based on fragmented activation distribution.

\begin{definition}[Confusion]
\label{def:confusion}
The \emph{confused} predicate is built from conflict catalogs (Definition~\ref{def:conflict-catalog-constructor}) and subgraph coherence \(H_{\mathrm{coh}}\).
Let \(R^{\mathrm{wm}} \subseteq V_{W_t}\) be the task-relevant WM node set for a given task representation, and let \(E_{R^{\mathrm{wm}}}\) be the edges among them.
Let \(G^{\mathrm{conf}}_{R^{\mathrm{wm}}} = (R^{\mathrm{wm}}, E_{R^{\mathrm{wm}}})\) denote the induced task-relevant subgraph (salience set \(R_t \subseteq V_B\) is a separate symbol).
Let
\[
\mathfrak{C}_{\mathrm{conf}}^{\mathrm{intra}}(R^{\mathrm{wm}}) =
\mathrm{BuildConflictCatalog}(\mathcal{F}_B, \mathcal{S}_B, \mathrm{scope}_{\mathrm{intra}}(R^{\mathrm{wm}}), W_t, B, \Gamma_t),
\]
where \(\mathrm{scope}_{\mathrm{intra}}(R^{\mathrm{wm}})\) retains intra-WM pairs \((X,Y)\) with \(V_X \cup V_Y \subseteq R^{\mathrm{wm}}\).
Let
\[
\mathfrak{C}_{\mathrm{conf}}^{\mathrm{wb}}(R^{\mathrm{wm}}) =
\mathrm{BuildConflictCatalog}(\mathcal{F}_B, \mathcal{S}_B, \mathrm{scope}_{\mathrm{wb}}(R^{\mathrm{wm}}), W_t, B, \Gamma_t),
\]
where \(\mathrm{scope}_{\mathrm{wb}}(R^{\mathrm{wm}})\) retains WM--belief pairs \((X,Y)\) with \(V_X \cap R^{\mathrm{wm}} \neq \varnothing\).
The \emph{confused} predicate relative to \(R^{\mathrm{wm}}\) is
\[
\mathrm{Confused}(W_t; R^{\mathrm{wm}}) \iff \mathrm{Disconnected}(R^{\mathrm{wm}}) \lor \mathrm{Incoherent}(R^{\mathrm{wm}}) \lor \mathrm{BeliefConflict}(R^{\mathrm{wm}}),
\]
where
\[
\mathrm{Disconnected}(R^{\mathrm{wm}}) \iff (R^{\mathrm{wm}}, E_{R^{\mathrm{wm}}}) \text{ is disconnected},
\]
\[
\mathrm{Incoherent}(R^{\mathrm{wm}}) \iff H_{\mathrm{coh}}(G^{\mathrm{conf}}_{R^{\mathrm{wm}}}, \mathfrak{C}_{\mathrm{conf}}^{\mathrm{intra}}(R^{\mathrm{wm}}), W_t) < \theta_{\mathrm{coh}},
\]
\[
\mathrm{BeliefConflict}(R^{\mathrm{wm}}) \iff \kappa(\mathfrak{C}_{\mathrm{conf}}^{\mathrm{wb}}(R^{\mathrm{wm}}), W_t) > 0,
\]
and \(\theta_{\mathrm{coh}}\) is the coherence threshold in \(\Theta\) (Definition~\ref{def:parameter-bundle}).
The corresponding \emph{confusion indicator} is the \(0/1\) quantity \(\Phi_{\mathrm{conf}}(W_t; R^{\mathrm{wm}}) = 1\) if and only if \(\mathrm{Confused}(W_t; R^{\mathrm{wm}})\).
The three disjuncts are \emph{under-connection} (\(\mathrm{Disconnected}\)), \emph{intra-WM incoherence} (\(\mathrm{Incoherent}\)), and \emph{cross-memory conflict} (\(\mathrm{BeliefConflict}\)).
\end{definition}

\subsection{Sweller's cognitive load in graph terms}
\label{sec:cognitive-load}

Cognitive load is modeled as a three-part decomposition into intrinsic, extraneous, and germane components, each expressed as a graph-theoretic functional on \(B\) and \(W_t\).
In this framework, intrinsic, extraneous, and germane load are represented by distinct graph-theoretic diagnostics over task structure, working-memory structure, and update-induced restructuring.
The load constructs below are mapping-level translations of Sweller's intrinsic, extraneous, and germane load distinction \cite{Sweller1988,SwellerChandler1994,Sweller2010}.
We write \(L_{\mathrm{load}}\) for the family \(\{L_{\mathrm{int}}, L_{\mathrm{ext}}(W_t), L_{\mathrm{ger}}\}\); each member is a distinct scalar measure over task schemas, working-memory graphs, or belief-graph updates.

\begin{definition}[Cognitive load measures]
\label{def:cognitive-load-components}
The three load measures are scalar measures defined as follows.
\begin{enumerate}
    \item \textbf{Intrinsic load.} Given a task that requires representing a minimal schema \(S^\ast = (V^\ast, E^\ast)\),
    \[
    \Phi_{\mathrm{load\mathrm{-}int}}(S^\ast) = L_{\mathrm{int}} = |V^\ast| + \lambda |E^\ast|,
    \]
    where \(\lambda \geq 0\) weights edge complexity.
    \item \textbf{Extraneous load.} For a working-memory state \(W_t\) during task execution, with task-relevant belief nodes \(R_t \subseteq V_B\) (Definition~\ref{def:salience-set}) and \(E^{\mathrm{sal}}_t \subseteq E_{W_t}\) the edges with both endpoints in \(\mathrm{active}(R_t, W_t)\),
    \[
    \Phi_{\mathrm{load\mathrm{-}ext}}(W_t) = L_{\mathrm{ext}}(W_t) = |V_{W_t} \setminus R_t| + \lambda |E_{W_t} \setminus E^{\mathrm{sal}}_t|,
    \]
    counting unnecessary nodes and edges, where \(E^{\mathrm{sal}}_t\) is distinct from confusion edge sets \(E_{R^{\mathrm{wm}}}\) (Definition~\ref{def:confusion}).
    \item \textbf{Germane load.} The germane load measure \(L_{\mathrm{ger}}\) quantifies coherence-preserving restructuring across an update \(B \to B'\), for example via \(\delta_{\mathrm{coh}}(B,B')\).
    For belief graphs \(B\) and \(B'\) produced during a sequence \((W_t, \dots, W_{t+n})\),
    \[
    \Phi_{\mathrm{load\mathrm{-}ger}}(B, B') = L_{\mathrm{ger}} = \delta_{\mathrm{coh}}(B, B'),
    \]
    where \(\delta_{\mathrm{coh}}\) quantifies the increase in structural coherence, e.g., via reductions in contradictions or increases in alignment between schemas and mental models.
\end{enumerate}
\end{definition}

\nestfigloadcaption[Intrinsic load \(L_{\mathrm{int}}=|V^\ast|+\lambda|E^\ast|\) for minimal task schemas \(S^\ast=(V^\ast,E^\ast)\) of increasing size and interconnectivity (left to right: \(|V^\ast|=2,|E^\ast|=1\); \(|V^\ast|=3,|E^\ast|=2\); \(|V^\ast|=4,|E^\ast|=5\)). Highlighted nodes and edges denote the task-relevant subgraph that must be coordinated in working memory.]{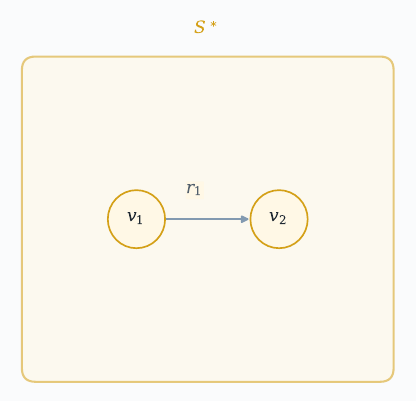}{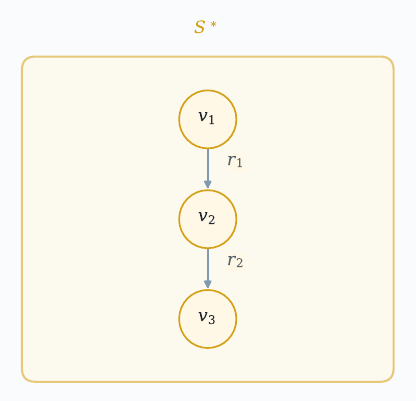}{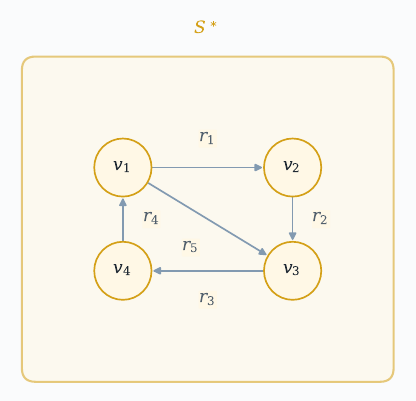}{fig06_intrinsic_load}

Figure~\ref{fig:fig06_intrinsic_load} illustrates intrinsic load over three minimal task schemas of increasing size and interconnectivity.

The intrinsic measure \(\Phi_{\mathrm{load\mathrm{-}int}}\) captures the inherent complexity of the minimal task schema \(S^\ast\): how many task-relevant elements and relations must be simultaneously maintained and coordinated in working memory.
Tasks that require many mutually constraining relations (e.g., multi-step algebra, multi-body physics problems) naturally induce higher intrinsic load because the corresponding subgraph is larger and more densely connected, even when instruction is optimal \cite{Sweller1988,SwellerChandler1994,Sweller2010}.

\begin{definition}[Effective element interactivity]
\label{def:effective-element-interactivity}
Let \(R_t \subseteq V_B\) be the task-relevant salience set (Definition~\ref{def:salience-set}) for a given task at time \(t\).
For a learner with belief graph \(B\), the \emph{effective element interactivity} scalar is
\[
E_{\mathrm{eff}}(R_t, B) = \min \bigl\{ k : R_t \subseteq \bigcup_{i=1}^{k} S^{\mathrm{ch}}_i,\ S^{\mathrm{ch}}_i \subseteq V_B\ \text{is a stored schema or chunk in } B \bigr\},
\]
i.e., the minimal number of stored chunks needed to cover the task-relevant elements.
\end{definition}

The quantity \(E_{\mathrm{eff}}(R_t, B)\) modulates how schema structure in \(B\) affects experienced intrinsic load: for a novice with few or poorly organized schemas, many task elements remain separate, so \(E_{\mathrm{eff}}\) is large; for an expert whose belief graph contains larger, well-integrated chunks that cover most of \(R_t\), \(E_{\mathrm{eff}}\) is smaller and the same task imposes lower intrinsic load \cite{Sweller1988,SwellerChandler1994,Sweller2010}.

\nestfigloadcaption[Extraneous load \(L_{\mathrm{ext}}(W_t)=|V_{W_t}\setminus R_t|+\lambda|E_{W_t}\setminus E^{\mathrm{sal}}_t|\) in working memory \(W_t\): green nodes lie in the task-relevant set \(R_t\); muted nodes and edges contribute avoidable complexity (left to right: \(|V_{W_t}\setminus R_t|=0\); \(=1\); \(=3\) with a separate extraneous subgraph).]{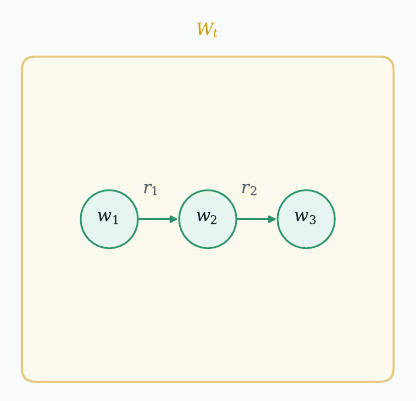}{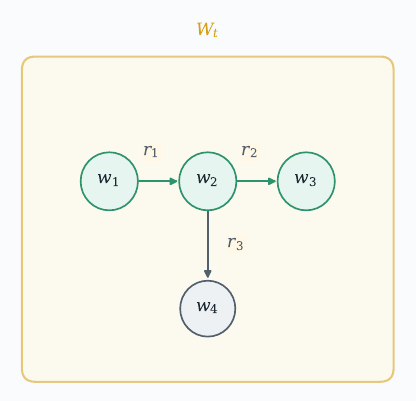}{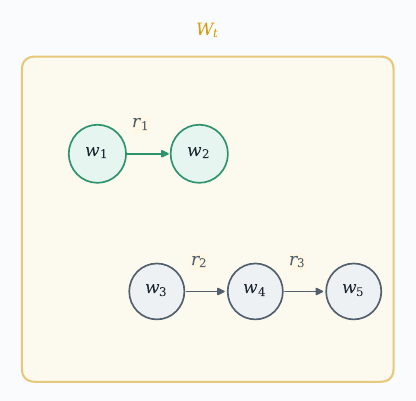}{fig07_extraneous_load}

Figure~\ref{fig:fig07_extraneous_load} shows extraneous load in three working-memory states with increasing avoidable complexity.

Extraneous load measures the portion of working-memory complexity that is not strictly required by the task as defined.
In graph terms, it corresponds to nodes and edges that are active in \(W_t\) but lie outside the ideal task-relevant subgraph, or that reflect unnecessarily convoluted representations of the same content (e.g., split attention, redundant but unintegrated information, confusing layouts).
Instructional designs that scatter related information across multiple components, use misleading examples, or introduce irrelevant details increase \(L_{\mathrm{ext}}(W_t)\) by activating and maintaining extra structure that competes with the task-relevant graph \cite{SwellerChandler1994,Sweller2010}.

\nestfigloadcaption[Germane load \(L_{\mathrm{ger}}=\delta_{\mathrm{coh}}(B,B')\): belief graph \(B\) updated to \(B'\) via productive restructuring (left to right: chunk formation by adding internal links among previously isolated nodes; integration via a new cross-link; alignment by bridging a mental-model subset \(M=\{v_1,v_2\}\) into the wider belief graph).]{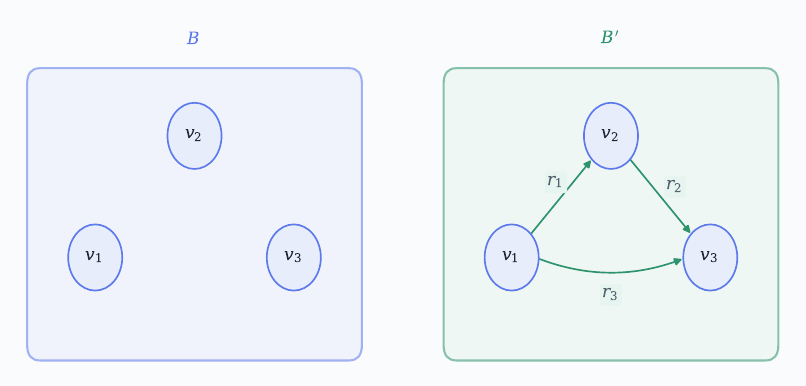}{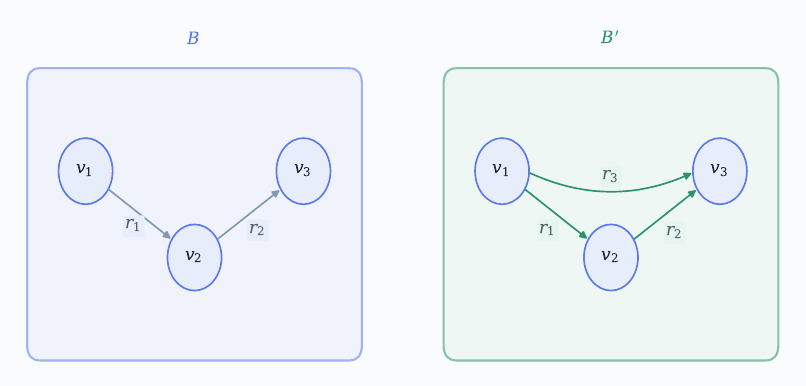}{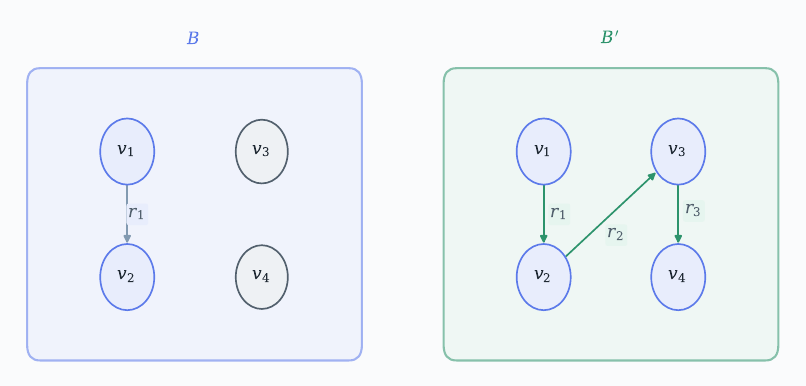}{fig08_germane_load}

Figure~\ref{fig:fig08_germane_load} shows germane load as three belief-graph updates \(B \to B'\) with increasing productive restructuring.

Germane load captures the portion of cognitive effort productively invested in schema construction and refinement rather than in coping with task complexity or poor presentation.
In the present formalism, germane load corresponds to graph operations that expand, reorganize, or integrate the belief graph in useful ways---for example, forming new chunks, adding explanatory causal links, or compressing multiple instances into a more general schema.
Thus, two states with similar total activation may differ in germane load: one may devote most of its resources to maintaining a large but static task subgraph, whereas another actively invokes belief-update operators \(\Delta_{\bullet}\) that improve future efficiency \cite{Sweller2010}.

Taken together, Definition~\ref{def:cognitive-load-components} recasts the standard Cognitive Load Theory distinction as parallel scalar measures: \(\Phi_{\mathrm{load\mathrm{-}int}}\) tracks inherent task-schema complexity, \(E_{\mathrm{eff}}(R_t,B)\) tracks how chunk structure in \(B\) modulates that complexity for a given learner, \(\Phi_{\mathrm{load\mathrm{-}ext}}\) tracks avoidable complexity in working memory, and \(\Phi_{\mathrm{load\mathrm{-}ger}}\) tracks productive belief-graph restructuring.
This makes it possible, in principle, to derive load-sensitive predictions from structural properties of \(W_t\) and from productive belief-graph restructuring within NEST.

\begin{proposition}[Trade-off constraint]
Assuming bounded working-memory capacity, increases in \(L_{\mathrm{ext}}\) reduce the available capacity for intrinsic and germane processing. Formally, for fixed WM capacity \(C_{\mathrm{WM}}\),
\[
L_{\mathrm{int}} + L_{\mathrm{ext}} + L_{\mathrm{ger}} \leq C_{\mathrm{WM}}.
\]
\end{proposition}

\begin{proof}
This follows directly from modeling load measures as counts or weighted counts of nodes and edges within a capacity-limited working-memory graph.
\end{proof}

Though simple, this formalization provides a basis for more refined theorems about how graph complexity interacts with performance and learning.

\subsection{Conceptual change as graph transformation}
\label{sec:conceptual-change}

\begin{definition}[Conceptual change]
\label{def:conceptual-change-operator}
\emph{Conceptual change} is a mapping-level construct: an instance of \(\mathrm{ApplyMacroUpdate}\) from the generic belief-update operator (Definition~\ref{def:belief-update-operator}) acting on the belief graph \(B\) in response to a sequence of working-memory states \((W_{t}, \dots, W_{t+n})\).
It restructures theory-bearing subgraphs of \(B\), potentially changing node labels, edge types, and compatibility relations.
In this application---written \(\Delta\) when the interpretation is conceptual change---\(\Delta\) may realize node addition or deletion, edge addition, deletion, or retyping, and reconfiguration of representational payloads, as already permitted under the generic belief-update framework.
The specialization lies in targeting theory-bearing subgraphs of \(B\) (including mental models and nested internal graphs \(G_v\)) during learning and theory revision \cite{Chi1992,diSessa2014}.
Conceptual change is distinguished from other belief-update applications such as self-update \(\Delta_{\mathrm{self}}\) (Section~\ref{sec:self}) by target subgraph; psychologically, it corresponds to theory revision and restructuring during learning.
Applications may be \emph{micro}---modifying a proper subgraph while preserving global topology---or \emph{macro}---reorganizing entire mental models or substantially reconfiguring a designated subgraph (Definition~\ref{def:belief-update-operator}).
Productive belief-graph restructuring tracked by germane load \(L_{\mathrm{ger}}\) (Section~\ref{sec:cognitive-load}) is realized under this interpretation as conceptual change.
\end{definition}

\subsection{Mental states as graph configurations}
\label{sec:mental-states}

Each mental state below is a trajectory-level or predicate-level diagnostic over working-memory graphs, belief graphs, awareness functional \(\mathcal{A}\), control functional \(\mathcal{K}\), conflict catalogs \(\mathfrak{C}_\kappa\), and belief-update operators \(\Delta_{\bullet}\) (Definition~\ref{def:belief-update-operator}).
The propositions identify recurring cognitive profiles as structured patterns within the formalism developed above; they characterize states rather than defining new primitives.

\begin{proposition}[Focused task engagement]
\label{prop:mental-state-focused}
A focused task state is typically characterized by a working-memory configuration in which fragmentation is low, typically \(F(W_t)=1\) (Definition~\ref{def:fragmentation}); the task-relevant salience set \(R_t\) (Definition~\ref{def:salience-set}) dominates \(\mathrm{active}(R_t,W_t)\) within \(W_t\); and activation mass \(\alpha_t^{W}\) is concentrated on one or a small number of task-relevant subgraphs.
Such states are additionally characterized by low extraneous load \(L_{\mathrm{ext}}(W_t)\) (Definition~\ref{def:cognitive-load-components}), low distraction measure \(\Phi_{\mathrm{dist}}(W_t)\) (Definition~\ref{def:distraction-measure}), and control outputs from \(\mathcal{K}\) (Definition~\ref{def:control-functional}) that preferentially stabilize task-relevant structure.
\end{proposition}

\begin{proposition}[Mind-wandering]
\label{prop:mental-state-mind-wandering}
Mind-wandering is typically characterized by a configuration in which \(W_t\) contains multiple weakly integrated or weakly related components, so that fragmentation \(F(W_t)\) is elevated (Figure~\ref{fig:fig05_fragmentation}) and overlap between \(V_{W_t}\) and \(\mathrm{active}(R_t,W_t)\) is reduced for the task-relevant salience set \(R_t\).
Activation mass is distributed across components, yielding increased distraction as measured by \(\Phi_{\mathrm{dist}}(W_t)\), while \(\mathcal{A}(W_t)\) and \(\mathcal{K}\) may remain confined to only a proper subgraph of currently active content.
\end{proposition}

\begin{proposition}[Rumination and perseveration]
\label{prop:mental-state-rumination}
Rumination is a trajectory-level phenomenon characterized by repeated reactivation of overlapping subgraphs with little productive belief-graph restructuring.
It may be further characterized by a sequence \(\{W_t\}\) in which successive working-memory states exhibit high mutual graph similarity (Definition~\ref{def:graph-similarity}) under the update operators \(\Delta_{\bullet}\); perseveration is the corresponding persistence of a narrow region of activation or control despite weak progress toward task completion.
\end{proposition}

\begin{proposition}[Insight and restructuring]
\label{prop:mental-state-insight}
\emph{Insight} is a trajectory-level diagnostic with two formal levels. At the working-memory level, insightful perception is the binary predicate of Definition~\ref{def:insightful-perception} (awareness of a WM-active subgraph together with temporal, conflict, evaluative, or provenance annotation); at the belief-graph level, restructuring corresponds to a comparatively large-scale application of the generic belief-update machinery---in particular, a macro-scale instance of \(\Delta\) interpreted as conceptual change (Definition~\ref{def:conceptual-change-operator})---yielding a more coherent task-relevant organization than in preceding states.
\end{proposition}

\begin{proposition}[Doubt and overconfidence]
\label{prop:mental-state-confidence}
Confidence-related states are binary predicates relating graph structure to control. \emph{Doubt} holds when conflict, fragmentation, distributed activation mass, or reduced subgraph coherence \(H_{\mathrm{coh}}\) are present and are reflected in regulation or hesitation by \(\mathcal{K}\) \cite{NelsonNarens1990,Koriat2007}; \emph{overconfidence} holds when control policies or evaluative outputs remain strong despite structural indicators of instability, contradiction, or poor integration within \(W_t\) or between \(W_t\) and \(B\).
\end{proposition}

\begin{proposition}[Attachment-maintenance state]
\label{prop:attachment-maintenance-state}
A working-memory state \(W_t\) exhibits attachment-maintenance when nodes from a designated attachment subgraph \(A_{\mathrm{att}} \subseteq B\) remain preferentially active or retrievable across time, due to attachment-specific salience parametrization (Remark~\ref{def:attachment-salience-param}) and corresponding control operations emitted by \(\mathcal{K}\).
\end{proposition}

\begin{proposition}[Fluent skill state]
\label{prop:fluent-skill-state}
A working-memory state exhibits relatively fluent or automatized skill when one or more procedurally compiled subgraphs (Remark~\ref{def:procedurally-compiled-subgraph}) are active and task performance proceeds with comparatively low explicit intermediate expansion, low fragmentation, and reduced dependence on \(\mathcal{K}\)-mediated regulation.
\end{proposition}

The characterizations above are not intended as a complete taxonomy, but they illustrate how diverse mental profiles may be expressed as derived graph configurations and trajectory patterns within the NEST architecture.

The characterizations above complete the phenomena layer without extending the foundational ontology. The next section specifies a reusable task-instantiation schema showing how these primitives may be organized into task-specific context graphs for expert reasoning, planning, and uncertainty-sensitive decision making.

\section{Task Instantiation Schema}
\label{sec:task-instantiation}

Task instantiation is expressed in terms of the belief graph \(B\), working-memory graph \(W_t\), salience sets \(R_t\) (Definition~\ref{def:salience-set}), grounding correspondences and the WM--belief collation graph \(G^{\mathrm{col}}_t = \mathrm{Coll}(W_t, B, \Gamma_t)\) (Definition~\ref{def:wm-belief-collation}), incompatibility constraints \((\mathcal{F}_B, \mathcal{S}_B)\) and cross-memory incompatibility (Definition~\ref{def:cross-memory-incompatibility}), conflict catalogs and active conflict count \(\kappa(\mathfrak{C}_\kappa, W_t)\) (Definition~\ref{def:active-conflict-count}), awareness subgraphs \(\mathcal{A}(W_t)\) (Definition~\ref{def:awareness-functional}), graph similarity \(\mathrm{sim}_{\mathrm{match}}\) (Definition~\ref{def:graph-similarity}), and control operations \(c \in \mathfrak{O}\) emitted by \(\mathcal{K}\) (Definition~\ref{def:control-functional}). A domain or field is represented as a mental model or designated region of \(B\); an expert is the bearer of a possibly fragmented subgraph of that region; and a question is a query subgraph whose resolution depends on similarity, grounding, and conflict relations over structures already defined in the formal core.

\subsection{Task-relative domain structure}

A task is always relative to some domain structure. That structure is given by a designated subgraph of \(B\), or by a connected mental model \(M \subseteq B\) (Definition~\ref{def:mental-model}). When a task concerns a broader field with multiple branches, the relevant domain may be represented by a family of overlapping mental models \(\mathcal{M}_{\mathrm{fld}} = \{M_1,\ldots,M_k\}\) together with their induced union subgraph inside \(B\); each \(M_i\) is already a subgraph of \(B\), and overlap among mental models is already permitted (Definition~\ref{def:mental-model}).
A task selects a relevant subgraph of the domain structure and constrains which graph operations, update paths, and answer statuses are admissible.

The domain structure supplies the background graph against which the current task is interpreted. Candidate actions, consequences, constraints, hypotheses, and failure modes are meaningful only relative to such a designated belief-graph region, because edge types, compatibility structure, and stored schemas needed for evaluation are inherited from \(B\). In this sense the full field knowledge resides in \(B\); task instantiation selects which region of \(B\) is relevant and which portion of that region the present agent can access.

\subsection{Field and expert-accessible subgraphs}

\begin{definition}[Field and expert-accessible subgraph]
\label{def:field-expert-subgraph}
Let \(B\) be the belief graph. A \emph{field} (or domain) is a designated connected mental model
\[
M = (V_M, E_M) \subseteq B
\]
(Definition~\ref{def:mental-model}). A particular expert's accessible knowledge is represented by a possibly fragmented subgraph
\[
M_{\mathrm{e}} = (V_{\mathrm{e}}, E_{\mathrm{e}}) \subseteq M,
\]
where \(M_{\mathrm{e}}\) captures the expert's internally available portion of the field. When \(M_{\mathrm{e}}\) is fragmented, write \(M_{\mathrm{e},1},\ldots,M_{\mathrm{e},k}\) for its connected components (ignoring edge direction, as in Definition~\ref{def:fragmentation}).
\end{definition}

Expert knowledge may be distributed across multiple weakly connected or disconnected subgraphs; coherence is task-relative, not a prerequisite for expertise.

\subsection{Questions as query subgraphs}

\begin{definition}[Question as query subgraph]
\label{def:question-query-subgraph}
A \emph{question} is a query subgraph
\[
G_{\mathrm{q}} = (V_{\mathrm{q}}, E_{\mathrm{q}}) \subseteq W_t
\]
active in working memory at time \(t\) (Definition~\ref{def:graph-similarity}). The subgraph \(G_{\mathrm{q}}\) may encode an incomplete relation, a candidate proposition, a required successor state, a predicted outcome, or a true/false alternative whose status is to be evaluated relative to the expert-accessible portion of the domain. Its resolution is evaluated through grounding correspondence \(\Gamma_t\) and \(G^{\mathrm{col}}_t\), partial matches \(f_{\mathrm{match}} : V_{\mathrm{q}} \to V_{M_{\mathrm{e}}}\), similarity scores \(\mathrm{sim}_{\mathrm{match}}(G_{\mathrm{q}}, M_{\mathrm{e}}; f_{\mathrm{match}})\), cross-memory incompatibility \(X \bowtie Y\) in \(G^{\mathrm{col}}_t\), active conflict count \(\kappa(\mathfrak{C}_{\mathrm{wb}}(W_t, B), W_t)\), and control outputs from \(\mathcal{K}\) that stabilize or revise the active query neighborhood in \(W_t\).
\end{definition}

A question is a structured probe into the domain-relevant belief region already represented in \(B\), not an isolated sentence. When the query is part of the active task neighborhood, \(G_{\mathrm{q}} \subseteq C_t \subseteq W_t\) for the context graph \(C_t\) defined below.

\subsection{Task context graphs}

\begin{definition}[Context graph]
\label{def:context-graph}
Let \(W_t\) be a working-memory graph at time \(t\), let \(R_t\) be a task-relevant salience set, let \(\mathcal{A}(W_t)\) be the subgraph currently in awareness, and let \(\mathcal{K}\) emit control operations \(c \in \mathfrak{O}\) that select and stabilize working-memory structure. A \emph{context graph} at time \(t\), denoted \(C_t\), is the active task-relevant subgraph of \(W_t\) selected under salience, awareness, and control:
\[
C_t = (V_{C_t}, E_{C_t}) \subseteq W_t,
\]
such that
\[
V_{C_t} \subseteq R_t \cup V_{\mathcal{A}(W_t)}.
\]
When useful, \(C_t\) may be defined as the maximal subgraph of \(W_t\) satisfying a task context predicate \(\Phi^{\mathrm{ctx}}_T\) (distinct from graph-property functional \(\Phi\), Definition~\ref{def:graph-property-functional}):
\[
C_t = \max\{X \subseteq W_t : \Phi^{\mathrm{ctx}}_T(X,t)=1\}.
\]
The evolution of \(C_t\) is governed by the working-memory transition operator \(U\), modulated by salience, awareness, and control; its relation to stored knowledge is mediated through \(\Gamma_t\) and \(G^{\mathrm{col}}_t\).
\end{definition}

\begin{remark}[Joint context as context-graph specialization]
\label{rem:joint-context-specialization}
Joint attention and shared intentionality do not require a new primitive graph type. They are represented by a socially indexed specialization of the context graph \(C_t\), in which the active context contains both self-related and other-agent-related nodes linked to a common target, event, or goal neighborhood. When a symbol is needed, write \(C_t^{\mathrm{joint}} \subseteq C_t\).
\end{remark}

\nestfigwithlabel[Context graph \(C_t\) with action branches, outcomes, and failure-mode nodes.]{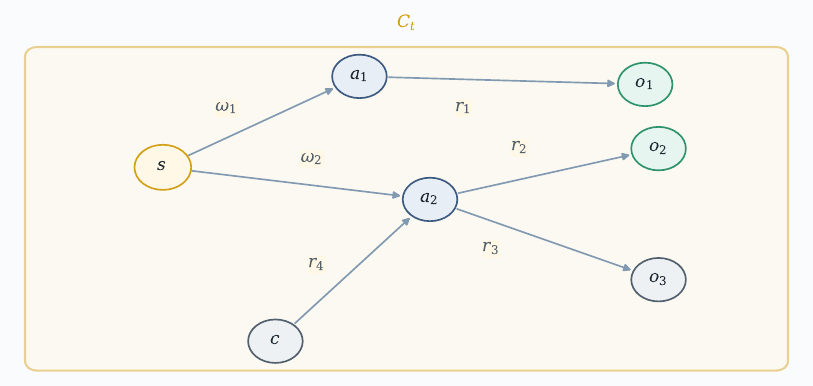}{fig09_context_graph}

Figure~\ref{fig:fig09_context_graph} illustrates a context graph \(C_t \subseteq W_t\) with candidate action branches, predicted outcomes, and failure-mode nodes realized as designated subgraphs within the base ontology. The context graph may include candidate action nodes, predicted outcome nodes, constraint nodes, goal nodes, and anticipated failure-mode nodes, all represented with the node and edge resources already available in the formalism.

\subsection{Answerability}

At time \(t\), the agent evaluates \(G_{\mathrm{q}}\) relative to \(M_{\mathrm{e}} \subseteq M \subseteq B\) using only operators and predicates already defined in the formal core.

\begin{definition}[Answerability]
\label{def:answerability}
Relative to expert-accessible subgraph \(M_{\mathrm{e}}\) and query subgraph \(G_{\mathrm{q}} \subseteq W_t\), the \emph{answer status} of \(G_{\mathrm{q}}\) is determined as follows:
\begin{enumerate}
    \item \textbf{Supported} when partial matches \(f_{\mathrm{match}} : V_{\mathrm{q}} \to V_{M_{\mathrm{e}}}\) yield high \(\mathrm{sim}_{\mathrm{match}}(G_{\mathrm{q}}, M_{\mathrm{e}}; f_{\mathrm{match}})\) (Definition~\ref{def:graph-similarity}), grounding scores \(\sigma_t(v,u)\) exceed \(\theta_{\mathrm{gnd}}\) in \(G^{\mathrm{col}}_t\) (Definition~\ref{def:wm-belief-collation}), and collation remains well-formed under \((\mathcal{F}_B, \mathcal{S}_B)\).

    \item \textbf{Contradicted} when \(X \subseteq G_{\mathrm{q}}\) and \(Y \subseteq M_{\mathrm{e}}\) satisfy \(X \bowtie Y\) under Definition~\ref{def:cross-memory-incompatibility} in \(G^{\mathrm{col}}_t\) (cf.\ failed recognition in Definition~\ref{def:recognition-event}), contributing to \(\kappa(\mathfrak{C}_{\mathrm{wb}}(W_t, B), W_t)\) (Proposition~\ref{prop:wm-belief-conflict-catalog}).

    \item \textbf{Unresolved} when the solution path requires nodes or edges in \(M \setminus M_{\mathrm{e}}\), spans a gap between disconnected components \(M_{\mathrm{e},i}\) and \(M_{\mathrm{e},j}\) of \(M_{\mathrm{e}}\), or remains undecided because grounding, coverage, or active conflict \(\kappa(\mathfrak{C}_\kappa, W_t)>0\) leaves the query open despite retrieval and control operations from \(\mathcal{K}\).
\end{enumerate}
\end{definition}

\begin{proposition}[Subset-based expertise]
\label{prop:subset-expertise}
Let \(M \subseteq B\) be a field mental model and \(M_{\mathrm{e}} \subseteq M\) the expert-accessible subgraph, with connected components \(M_{\mathrm{e},1},\ldots,M_{\mathrm{e},k}\) when \(M_{\mathrm{e}}\) is fragmented. Questions whose resolution paths lie within a single component \(M_{\mathrm{e},i}\) are within direct competence on that region; questions requiring structure in \(M \setminus M_{\mathrm{e}}\), or paths that bridge gaps between components of \(M_{\mathrm{e}}\), are outside direct competence. Expertise is access to a subgraph of a designated field region---possibly distributed across several coherent components---rather than possession of all of \(B\).
\end{proposition}

Different experts correspond to different subgraphs \(M_{\mathrm{e}}\). Reasoning over a question is graph traversal plus compatibility checking---grounding, similarity, conflict resolution, and control---relative to the subgraph the agent can currently activate in \(W_t\) and retrieve from \(M_{\mathrm{e}}\).

\subsection{True/false evaluation and branch knowledge}

True/false judgments illustrate the schema directly. A true/false question may be modeled as a query subgraph \(G_{\mathrm{q}}\) containing a candidate node or relation together with the structural commitments required for support or contradiction. Evaluation consists in testing whether the relevant support, contradiction, temporal ordering, causal dependence, or containment structure can be recovered from \(M_{\mathrm{e}}\) and its currently activated neighborhood in \(W_t\).

An expert is therefore not a store of detached answers, but an agent who knows enough branches of \(M\) to traverse from \(G_{\mathrm{q}}\) to supporting or contradicting structure within \(M_{\mathrm{e}}\). Different experts may answer the same question differently when their accessible subgraphs differ in coverage, integration, or current activation, even though the full field \(M\) may contain a definitive answer path.

\subsection{Action branches and failure modes}

The same context graph \(C_t\) can encode decision structure in addition to question structure. Candidate actions are represented as nodes or subgraphs linked by temporal, causal, evidential, and evaluative edges to predicted outcomes, intermediate states, constraints, and failure modes. The question of what to do next is a structured comparison among successor branches internal to \(C_t\), rather than an unstructured choice among labels.

Failure modes are designated subgraphs whose activation or predicted realization conflicts with goal structure, constraint satisfaction, or coherence conditions over the task-relevant region. A control operation \(c \in \mathfrak{O}\) emitted by \(\mathcal{K}\) may stabilize one branch, suppress another, retrieve missing structure from \(B\), or route processing toward conflict resolution (Definition~\ref{def:control-operation}).

\subsection{Task specification}

\begin{definition}[Task specification]
\label{def:task-specification}
A task specification is a tuple
\[
T = \langle R_T,\Gamma_T,\Omega_T,\Theta_T\rangle,
\]
where:
\begin{enumerate}
    \item \(R_T \subseteq V_B\) is the specification-level task-relevant salience set, and \(R_t\) is its runtime instantiation (Definition~\ref{def:salience-set});
    \item \(\Gamma_T\) is a set of admissibility constraints;
    \item \(\Omega_T\) is a set of candidate operations or action options;
    \item \(\Theta_T\) is a set of task predicates, including success and failure conditions \(\Psi^{\mathrm{task}}_{\mathrm{succ}}\) and \(\Psi^{\mathrm{task}}_{\mathrm{fail}}\) (distinct from trajectory functional \(\Psi\), Definition~\ref{def:trajectory-functional}) and from the global threshold bundle \(\Theta\) (Definition~\ref{def:parameter-bundle}).
\end{enumerate}
\end{definition}

The tuple \(T\) is a task-indexed specification of selection, admissibility, and outcome structure. For a given task \(T\), the corresponding context graph is written \(C_t^T\), where
\[
C_t^T \subseteq W_t
\qquad\text{and}\qquad
V_{C_t^T} \subseteq R_t \cup V_{\mathcal{A}(W_t)}.
\]
At runtime, the task-relevant set \(R_t\) is the instantiated slice of \(R_T\) present in the current state.
Equivalently, the task determines a projection
\[
\Lambda_T : W_t \mapsto C_t^T
\]
from the current working-memory state to the task-relevant context graph.

\subsection{Goals, constraints, and options}

Within a task-instantiated context graph, active goals are represented as designated nodes or subgraphs already licensed by the base ontology. Write \(V^{\mathrm{goal}}_t \subseteq V_{C_t}\) for the set of currently active goal nodes. Constraints are drawn from the admissibility family \(\Gamma_T\) in the task specification (Definition~\ref{def:task-specification}), instantiated as admissibility predicates \(\{\gamma^{\mathrm{adm}}_1,\ldots,\gamma^{\mathrm{adm}}_m\}\) (lowercase \(\gamma\), not grounding or task-family \(\Gamma\)); grounding correspondence \(\Gamma_t\) (Definition~\ref{def:wm-belief-collation}) is a separate object. Candidate options are represented by a set \(\Omega_t = \{\nu_1,\ldots,\nu_n\}\) instantiated from \(\Omega_T\), where each \(\nu_i\) is either a control-eligible operation in \(\mathfrak{O}\) or a structured transition over working memory induced by \(U\). An option \(\nu \in \Omega_t\) is admissible at time \(t\) iff \(\mathrm{Adm}(\nu,C_t,\Gamma_T)=1\).

\begin{definition}[Goal-designated subgraph]
\label{def:goal-designated-subgraph}
A motivational or task goal state is represented by a designated subgraph
\[
G_t=(V_{G_t},E_{G_t}) \subseteq B,
\]
whose nodes encode desired outcomes, maintained constraints, or preferred target states. This is not a new primitive object class: it is a designated belief subgraph whose task-layer runtime realization may appear inside the context graph \(C_t\).
\end{definition}

\begin{definition}[Goal satisfaction]
\label{def:goal-satisfaction}
Let \(G_t=(V_{G_t},E_{G_t}) \subseteq B\) be a designated goal subgraph. Given a realization predicate \(\mathrm{realized}(g,W_t)\) on goal nodes \(g \in V_{G_t}\), the \emph{goal-satisfaction measure} is the involvement-style quantity
\[
S_{\mathrm{goal}}(G_t,W_t)
=
\frac{|\{g \in V_{G_t} : \mathrm{realized}(g,W_t)\}|}{|V_{G_t}|}.
\]
When realization is identified with working-memory activation or explicit attainment inside a context graph, this is an application-specific specialization of Definition~\ref{def:subgraph-involvement}.
\end{definition}

\begin{remark}[Planning as prospective subgraph]
\label{rem:planning-prospective-subgraph}
For a goal subgraph \(G_t\), a \emph{planning graph} is a designated belief subgraph \(P^{\mathrm{plan}}_t = (V_{P^{\mathrm{plan}}_t}, E_{P^{\mathrm{plan}}_t}) \subseteq B\) containing current task-state nodes, candidate action branches (typically \(\mathcal{R}_{\mathrm{causal}}\)), intermediate states linked by temporal edges (\(\mathcal{R}_{\mathrm{temp}}\)), and goal nodes in \(V_{G_t}\), with directed paths from present states toward goals. This is a prospective action subgraph within \(B\), not a new primitive graph type.
\end{remark}

\begin{proposition}[Motivational profile]
\label{prop:motivational-profile}
A motivational state at time \(t\) may be characterized by the activation pattern of goal nodes in \(G_t\), the structure of available planning subgraphs directed toward those goals, and the expected utilities associated with action branches connecting current states to goal nodes. High motivation corresponds to strong goal activation, connected planning structure, and control allocation toward high-utility paths; apathy corresponds to the converse profile.
\end{proposition}

\subsection{Predicted outcomes}

For each admissible option \(\nu \in \Omega_t\), the predicted outcome is
\[
\widehat{W}_{t+1}^{(\nu)} = U^{\mathrm{tsk}}_T(W_t,\nu),
\]
where \(U^{\mathrm{tsk}}_T\) denotes the working-memory transition operator \(U\) (Definition~\ref{def:wm-transition}) as constrained by task specification \(T\) (distinct from chunking \(U_\chi\) and activation update \(U_{\mathrm{act}}\)). The set of predicted successor states is \(\mathcal{W}^{\mathrm{pred}}_t = \{\widehat{W}_{t+1}^{(\nu)} : \nu \in \Omega_t\}\). Option selection may then be defined separately by an objective functional \(J^{\mathrm{obj}}_T\), if the task requires ranking among predicted successors.

\subsection{Working-memory instantiation}

At runtime, the active task is realized as a working-memory graph \(W_t\) containing a task-relevant slice of \(M_{\mathrm{e}}\) and the current query subgraph \(G_{\mathrm{q}}\). Retrieval from \(M_{\mathrm{e}}\) into \(W_t\), grounding through \(\Gamma_t\), and control outputs from \(\mathcal{K}\) jointly determine which field structure is available for inference.

The context graph \(C_t\) is obtained from the interaction of \(G_{\mathrm{q}}\), the currently active portion of \(M_{\mathrm{e}}\) in \(W_t\), and salience, awareness, and control constraints as in Definition~\ref{def:context-graph}. Successful reasoning aligns the active neighborhood of \(G_{\mathrm{q}}\) with the relevant branch of \(M\); failure occurs when the needed branch is absent from \(M_{\mathrm{e}}\), weakly represented in \(W_t\), or in conflict under \(G^{\mathrm{col}}_t\). When \(G_{\mathrm{q}}\) is only weakly aligned to \(M_{\mathrm{e}}\), when grounding scores are low, or when conflict catalogs remain active, the architecture predicts hesitation, partial resolution, or controlled search rather than confident commitment.

\subsection{Success and failure conditions}

Task success and failure are represented by task predicates \(\Psi^{\mathrm{task}}_{\mathrm{succ}}(W_t,T)\) and \(\Psi^{\mathrm{task}}_{\mathrm{fail}}(W_t,T)\) over working-memory states.

\begin{definition}[Failure state]
\label{def:task-failure-state}
A working-memory state \(W_t\) is a \emph{failure state} relative to task \(T\) iff
\[
\mathrm{Fail}^{\mathrm{task}}(W_t,T)
\iff
\neg \Psi^{\mathrm{task}}_{\mathrm{succ}}(W_t,T)\wedge \Psi^{\mathrm{task}}_{\mathrm{fail}}(W_t,T).
\]
where \(\mathrm{Fail}^{\mathrm{task}}\) is distinct from the fragmentation index \(F(W_t)\) (Definition~\ref{def:fragmentation}).
\end{definition}

Typical failure instances, all expressed with previously defined quantities, include excessive fragmentation \(F(W_t)>1\) (Definition~\ref{def:fragmentation}), unresolved active conflict \(\kappa(\mathfrak{C}_\kappa,W_t)>0\), incoherence in a designated task-relevant subgraph, inadequate awareness coverage \(V_{C_t}\nsubseteq V_{\mathcal{A}(W_t)}\), and control trajectories that violate task constraints.

\subsection{Recovery and task-sensitive control}

When a failure state is detected, recovery is represented by a control operation \(c_{\mathrm{rec}} \in \mathfrak{O}\) emitted by \(\mathcal{K}\)---for example attentional refocusing, retrieval amplification, distractor suppression, persistence bias, or conflict-resolution routing (Definition~\ref{def:control-operation}). Task-level regulation is the task-relative deployment of \(\mathcal{K}\) over designated context graphs and their trajectories.

\subsection{Task-sensitive trajectories}

A task-sensitive trajectory is the sequence \((C_{t-k}^T,\ldots,C_t^T)\), where each \(C_{t-i}^T\) is induced from \(W_{t-i}\) by the task projection \(\Lambda_T\). Task performance depends on the sequence of transformations through which \(C_t^T\) arose. Those transformations may depend on fragmentation, conflict count, self-involvement, self-consistency, or signed evaluation along the trajectory.

\subsection{Summary}

The task-instantiation schema specializes previously defined graph operations over \(B\) and \(W_t\):
\begin{itemize}
    \item \(M \subseteq B\): a field or domain, designated as a connected mental model (Definition~\ref{def:mental-model}), or as a family \(\mathcal{M}_{\mathrm{fld}}\) of overlapping mental models for multi-branch fields;
    \item \(M_{\mathrm{e}} \subseteq M\): expert-accessible knowledge, a possibly fragmented subgraph with components \(M_{\mathrm{e},1},\ldots,M_{\mathrm{e},k}\) (Definition~\ref{def:field-expert-subgraph});
    \item \(G_{\mathrm{q}} \subseteq W_t\): the question, as a query subgraph and structured probe into \(M\) (Definition~\ref{def:question-query-subgraph});
    \item \(C_t \subseteq W_t\): the task context graph, selected under salience \(R_t\), awareness \(\mathcal{A}(W_t)\), and control from \(\mathcal{K}\) (Definition~\ref{def:context-graph});
    \item answer status: supported, contradicted, or unresolved under Definition~\ref{def:answerability}, via \(\mathrm{sim}_{\mathrm{match}}\), \(\Gamma_t\), \(G^{\mathrm{col}}_t\), \(\bowtie\), and \(\kappa\).
\end{itemize}
The query subgraph specifies what is asked, the context graph specifies what is currently relevant, and answer status specifies how the domain structure resolves the query.
Expert reasoning is the task-relative activation, alignment, and regulation of domain-specific subgraphs already available within \(B\) and \(W_t\).

\begin{remark}[Scope of the section]
\label{rem:task-instantiation-scope}
This section specifies a reusable task-instantiation schema. Context graphs \(C_t\), mental models \(M\), expert-accessible subgraphs \(M_{\mathrm{e}}\), query subgraphs \(G_{\mathrm{q}}\), goals, constraints, options, predicted outcomes, recovery policies, and failure states are task-indexed designations, predicates, or constrained uses of objects already defined in the foundational architecture. The difference between the present section and the formal core is one of interpretation and parameterization.
\end{remark}

\section{Mappings to Existing Frameworks}
\label{sec:mappings}

This section does not claim that existing frameworks are identical to NEST. Rather, it gives compatibility mappings: for each framework, we identify the subset of NEST variables, operators, and constraints needed to represent its central commitments. The comparison is therefore structural rather than eliminative: differences between frameworks appear as differences in which graph objects are primitive, which transitions are licensed, and which constraints are enforced. This approach follows the broader unification program in cognitive science, which treats architecture-level integration as a scientific goal rather than a mere re-description exercise \cite{Newell1990,Milkowski2017,LairdLebiereRosenbloom2017}.

\subsection{ACT-R}

ACT-R provides a natural starting point because it already distinguishes declarative knowledge, procedural knowledge, and buffer-based access to currently relevant information \cite{Anderson2005,LebiereACT,AndersonEtAl2004ACTR}. In NEST, declarative chunks map to recursive nodes in the belief graph \(B\), while ACT-R buffers correspond not to all of \(W_t\), but to functionally designated active subgraphs or slices of \(W_t\) selected for current processing. This preserves ACT-R's separation between stored chunk structure and transient buffer occupancy. The ACT-R goal buffer corresponds to a designated task-relevant subgraph of \(W_t\), whereas declarative retrieval corresponds to graph matching and activation-dependent access over nodes of \(B\). Productions correspond to constrained transition templates: condition-sensitive instances of the working-memory transition operator \(U\) together with possible belief-update consequences under \(\Delta\). A production matches a configuration of buffer contents and then issues such a transition on \(W_t\), on \(B\), or on the interface between them. This preserves ACT-R's separation between what is known, what is currently active, and what action is executed from that state \cite{Anderson2005,LebiereACT,AndersonEtAl2004ACTR}.

Conceptually, ACT-R's central insight is that cognition is controlled by the interaction of activation, retrieval, and rule-based action selection \cite{Anderson2005,AndersonEtAl2004ACTR}. NEST preserves that insight but makes the representational substrate explicit. Activation becomes \(\alpha_t^B\) and \(\alpha_t^{W}\), retrieval becomes grounded access from \(W_t\) to \(B\), and rule application becomes \(\Delta_{\bullet}\) or \(U\), depending on whether the transition revises belief structure or only changes the working-memory configuration. What ACT-R leaves implicit is the internal structure of chunks and the relational structure between chunks; NEST represents both explicitly through recursive nodes and typed edges. NEST adds explicit internal graph structure and typed relational organization while preserving ACT-R's commitment to activation-driven retrieval and production-based control \cite{Anderson2005,AndersonEtAl2004ACTR}.

\subsection{Soar}

Soar can be mapped even more directly because its architectural vocabulary already centers on problem spaces, operators, subgoals, impasses, and chunking \cite{Newell1990,Newell1989Soar,LairdRosenbloomNewell1986,Laird2012Soar}. In NEST, a Soar problem space corresponds to a task-context subgraph in \(W_t\), operators correspond to candidate graph transitions, and the selection of an operator corresponds to choosing one of several permissible update paths through the task graph. An impasse corresponds to a control-relevant failure state in which the current task-context graph does not support decisive operator selection; in NEST this may arise from incompatibility, incomplete task coverage, unresolved grounding, or insufficient discriminability among available successor transitions. Soar's chunking mechanism then maps to belief-update compression in NEST: repeated problem-solving traces can be compiled into new graph structure in \(B\), thereby reducing future search cost and increasing direct access to solved patterns \cite{Newell1990,LairdRosenbloomNewell1986,Laird2012Soar}.

Conceptually, Soar and NEST share the idea that intelligent behavior emerges from the interaction of a persistent knowledge base with a limited working state and a mechanism that resolves control uncertainty. NEST preserves Soar's universal subgoaling logic, but replaces Soar's more procedural description with explicit graph-theoretic diagnostics. In particular, subgoaling becomes the introduction of a derived task-context subgraph whose function is to resolve the currently blocked control state, while impasse corresponds to a control-relevant failure state in which the current task-context graph does not support decisive operator selection. This mapping is valuable because it shows that Soar's control cycle can be understood as a special case of graph transition under structural constraints, with chunking as a form of graph-theoretic memory consolidation \cite{Newell1990,Newell1989Soar,LairdRosenbloomNewell1986,Laird2012Soar}.

\subsection{Sigma and the Common Model of Cognition}

Sigma and the Common Model of Cognition (CMC) are related but not identical comparison targets. The CMC is primarily an architectural inventory of major cognitive subsystems, whereas Sigma is a more unified architectural framework with graded and probabilistic commitments. NEST can map to both, but the correspondence is clearest when the two are distinguished before comparison \cite{LairdLebiereRosenbloom2017,LairdLebiereRosenbloomStocco2025,Rosenbloom2010,RosenbloomDemskiUstun2016}.

At the CMC level, long-term memory maps to the belief graph \(B\), working memory to \(W_t\), control to \(\mathcal{K}\) and related trajectory functionals, perceptual intake to transient graph construction from perceptual payloads, action to graph-externalized transitions or task-selected successor states, and procedural organization to constrained instances of \(U\) and \(\Delta\). Sigma's hybrid and probabilistic orientation is compatible with NEST insofar as weighted edges, graded activation, grounding scores, and threshold-parameterized transitions allow uncertainty and partial support to be represented without abandoning explicit graph structure \cite{LairdLebiereRosenbloom2017,LairdLebiereRosenbloomStocco2025,RosenbloomDemskiUstun2016}.

Conceptually, the CMC contributes a strong modular view of cognition: information is perceived, integrated into working memory, processed by control, and then used to guide action and learning \cite{LairdLebiereRosenbloom2017}. NEST preserves this control flow but adds a finer-grained content model. The CMC specifies where information lives; NEST specifies what that information looks like internally and how it can be recursively nested. This distinction matters because it allows NEST to represent not only module-to-module communication, but also the internal conceptual, episodic, and relational structure of the information being communicated. In that sense, NEST can be read as a representational enrichment of the CMC rather than a replacement for it \cite{LairdLebiereRosenbloom2017,LairdLebiereRosenbloomStocco2025}.

\subsection{Global Workspace Theory}

Global Workspace Theory (GWT) and Global Neuronal Workspace models treat conscious access as a limited-capacity workspace in which selected information becomes globally available to multiple specialized processors through broadcasting or ignition-like dynamics \cite{Baars1988,DehaeneChangeux2011,Dehaene2014,DehaeneKerszbergChangeux1998}. NEST maps this idea by treating the workspace as an access-privileged, high-salience region of \(W_t\) whose contents are available for broad coordination across otherwise specialized subgraphs, while the global broadcast corresponds to activation spread from a focal subgraph into other active subgraphs connected through typed relations. A conscious content, under this mapping, is a working-memory subgraph that is both sufficiently activated and selected by the awareness functional \(\mathcal{A}(W_t)\) for broad control-relevant access. The workspace is therefore a control-sensitive, graph-structured access region rather than a module with no internal content \cite{Baars1988,DehaeneChangeux2011,Dehaene2014}.

Conceptually, GWT explains why some information becomes available for report, decision, and control while other information remains local or inaccessible \cite{Baars1988,DehaeneChangeux2011}. NEST preserves that access distinction by separating \(B\) from \(W_t\) and by defining awareness via the awareness functional \(\mathcal{A}(W_t)\) (Definition~\ref{def:awareness-functional}), but it extends GWT by making the workspace content explicitly structured and recursively nested. This extension is important because the workspace in NEST can contain not only flat propositions, but also graph payloads, task models, self-related structures, and conflict patterns. GWT therefore appears in NEST as a special case in which conscious access is modeled by the joint action of activation, awareness selection, and cross-subgraph availability within working memory \cite{Baars1988,DehaeneChangeux2011,DehaeneKerszbergChangeux1998,Dehaene2014}.

\subsection{Semantic Networks and Conceptual Graphs}

Semantic networks provide the simplest graph-level predecessor of NEST: concepts are nodes and relations such as \emph{is-a}, \emph{part-of}, and \emph{related-to} are edges \cite{CollinsQuillian1969,CollinsLoftus1975,Steyvers2005}. NEST maps these directly into its containment and associative edge classes, with recursive node payloads extending the notion of a concept node so that it can contain internal substructure rather than remain atomic. Thus, a semantic network constitutes a restricted representational case of NEST in which the graph has minimal nesting and the main operations involve spreading activation and relational access among concept nodes \cite{CollinsQuillian1969,CollinsLoftus1975,Steyvers2005}. The conceptual gain from NEST is that semantic relations are no longer the whole story; they become one relation class inside a broader typed ontology that also includes causal, temporal, evidential, and spatial structure.

Conceptual graphs strengthen semantic networks by providing a formal graph language for propositional and relational representation \cite{Sowa1984}. NEST preserves this representational ambition, but adds working-memory dynamics, belief-update operators, and explicit graph-based conflict conditions. In conceptual-graph terms, NEST can represent propositions as relational subgraphs, but it additionally specifies how those subgraphs become active, how they are checked against stored structure, and how they are revised over time. Conceptual graphs map into NEST as a special representational case: they capture richly structured relational content, while NEST adds explicit working-memory dynamics, conflict diagnostics, and update operators over that content \cite{Sowa1984,Steyvers2005}.

\subsection{Theory-Theory and chunking}

Theory-Theory approaches treat concepts as embedded in intuitive explanatory systems rather than as isolated labels \cite{GopnikMeltzoff1997,GopnikWellman1994,Murphy2002TheoryTheory}. In NEST, such intuitive theories are naturally modeled as recursive subgraphs containing causal, evidential, temporal, and containment relations among lower-level concepts. A concept is then not merely an atomic node, but a node whose payload may encode an explanatory microtheory as an internal graph. Conceptual change becomes graph revision under \(\Delta\): a local theory-structure may be reweighted, retyped, partially replaced, or reorganized when incompatible evidence accumulates or when a rival structure yields higher coherence. This provides a direct representational bridge between theory-based conceptual development and graph-based revision dynamics \cite{GopnikWellman1994,Murphy2002TheoryTheory}.

Chunking can be mapped to NEST as compression of recurrent graph patterns into stable belief-graph structures. Repeated problem-solving episodes can be collapsed into a new recursive node or reusable subgraph, which then functions as a higher-level unit of access in later processing.

\begin{definition}[Chunking operator]
\label{def:chunking-operator}
A \emph{chunking operator} is a two-stage operator family consisting of a working-memory compression step \(U_\chi\) and, when committed, a belief-update step \(\Delta_\chi\). The first stage compresses a set of co-active nodes in \(W_t\) into a single new recursive node; the second stage stores the resulting chunk in \(B\) when update conditions of the architecture are satisfied. Let \(\{v_1, \dots, v_k\} \subseteq V_{W_t}\) be co-active nodes---including perceptual input nodes---compressed into
\[
v^\ast = (c^\ast, G^\ast),
\]
where \(G^\ast\) is a graph payload encoding the internal structure of \(\{v_1, \dots, v_k\}\) and their mutual edges. The WM compression stage is
\[
U_\chi(W_t, B) = W_{t+1},
\]
where \(V_{W_{t+1}} = (V_{W_t} \setminus \{v_1, \dots, v_k\}) \cup \{v^\ast\}\), edges incident to any \(v_i\) are redirected to \(v^\ast\), and \(\alpha_{t+1}^{W}(v^\ast) = \sum_{i=1}^{k} \alpha_t^{W}(v_i)\). The operator reduces \(|V_{W_t}|\) by \(k-1\), freeing capacity. A node \(v^\ast\) formed this way is called a \emph{chunk}; if a matching chunk already exists in \(B\), the operation is a recognition event rather than a novel chunk-creation event (Definition~\ref{def:recognition-event}). When update conditions are satisfied, \(\Delta_\chi\) commits the chunk structure to \(B\).
\end{definition}

Accordingly, chunking is transition-level when treated as temporary compression inside \(W_t\), and update-level when the compressed structure is committed to \(B\).

This preserves the central intuition of chunking literature: repeated computation can be compiled into memory structures that reduce future processing cost. In NEST, however, the compiled object is not merely a symbolic chunk; it is a typed graph fragment with internal structure and update history \cite{GobetSimon1996,Newell1990}.

\subsection{Graph-Based AI and Knowledge-Graph Reasoning}

Graph-based AI provides a computationally relevant comparison class for NEST. Knowledge graphs and graph neural networks have become standard tools for representing explicit relational structure, improving interpretability, and supporting inference over structured data \cite{Pan2024Roadmap,Dai2024KG,Kau2024KG,Wu2021GNN}. Recent work shows that large language models can benefit from knowledge graphs as structured sources of relational grounding, and that graph-based representations can improve factual consistency, reasoning reliability, and explainability \cite{Pan2024Roadmap,Dai2024KG,Kau2024KG}. Graph neural networks are also increasingly used to model cognitive and brain-level structure, including connectome-based prediction and interpretable cognitive decoding \cite{Zhang2022Connectome,Zhang2022BGNN}.

NEST fits this direction because its primitives are already graph-native. Its belief graph, working-memory graph, typed edges, grounding correspondences, and update operators can be implemented in forms compatible with knowledge-graph systems and, in part, with graph-learning pipelines. Conceptually, this matters because AI systems that aim to model human cognition need not only distributed pattern learning, but also explicit structural representations of memory, conflict, control, and update. NEST provides an explicit representational scheme for those structures, with computational implementations available in graph-native settings, while remaining close to the representational vocabulary used in contemporary graph-based AI and cognitive modeling \cite{Pan2024Roadmap,Dai2024KG,Wu2021GNN,Zhang2022Connectome,Zhang2022BGNN}.

\subsection{Summary of correspondences}

Table~\ref{tab:framework-mappings} consolidates the variable-level mappings developed above.

\begin{table}[ht]
\centering
\caption{Summary of framework correspondences at the architectural-variable level.}
\label{tab:framework-mappings}
\small
\begin{tabular}{@{}p{0.18\linewidth}p{0.34\linewidth}p{0.38\linewidth}@{}}
\hline
\textbf{Framework} & \textbf{Source-level variables or modules} & \textbf{NEST correspondence} \\
\hline
ACT-R &
Chunks, buffers, productions, declarative memory, procedural memory &
Recursive nodes in \(B\); designated active subgraphs of \(W_t\); constrained instances of \(U\) and \(\Delta\); activation and retrieval as \(\alpha_t^B\), \(\alpha_t^{W}\), and grounded access \cite{Anderson2005,AndersonEtAl2004ACTR} \\
Soar &
Problem spaces, operators, impasses, chunking, subgoals &
Task-context subgraphs; candidate transitions; control-relevant failure states; derived subgoal contexts; compiled graph patterns in \(B\) \cite{Newell1990,Newell1989Soar,LairdRosenbloomNewell1986,Laird2012Soar} \\
CMC / Sigma &
Long-term memory, working memory, control, perception, action; graded/probabilistic processing &
\(B\), \(W_t\), \(\mathcal{K}\), \(\Psi\), constrained instances of \(U\) and \(\Delta\), perceptual graph input, action-selective transitions, and weighted or graded graph relations \cite{LairdLebiereRosenbloom2017,LairdLebiereRosenbloomStocco2025} \\
GWT &
Workspace, access, broadcast, ignition, conscious availability &
Access-privileged subgraph of \(W_t\) selected by \(\mathcal{A}(W_t)\); activation spread and broad availability across connected subgraphs \cite{Baars1988,DehaeneChangeux2011,Dehaene2014,DehaeneKerszbergChangeux1998} \\
Semantic networks / conceptual graphs &
Concepts, arcs, \emph{is-a}, \emph{part-of}, relational inference &
Containment and associative edges; recursively structured nodes; graph-based propositional content \cite{CollinsQuillian1969,CollinsLoftus1975,Sowa1984,Steyvers2005} \\
Theory-Theory / chunking &
Intuitive theories, explanatory structure, compiled knowledge &
Recursive theory subgraphs; revision under \(\Delta\); temporary chunk compression in \(W_t\); reusable chunk storage in \(B\) \cite{GopnikWellman1994,Murphy2002TheoryTheory,GobetSimon1996} \\
Graph-based AI &
Knowledge graphs, graph neural networks, structured reasoning &
Direct computational substrate for belief, working memory, conflict, grounding, and update \cite{Pan2024Roadmap,Dai2024KG,Kau2024KG,Wu2021GNN,Zhang2022Connectome,Zhang2022BGNN} \\
\hline
\end{tabular}
\end{table}

The table records correspondence at the level of architectural variables and operators, not identity of explanatory aims or empirical commitments.

The point of these mappings is not that every framework collapses into NEST without remainder. It is that the frameworks can be compared within a shared representational space. Once stated in that space, their differences become more explicit: they differ in which graph objects are treated as primitive, which transitions are licensed, which conflicts are representable, and which update policies are permitted \cite{Newell1990,Milkowski2017,LairdLebiereRosenbloom2017}.

\section{Discussion and Future Work}
\label{sec:discussion}

This paper is intentionally foundational. Its aim has been to specify the representational and structural commitments of NEST at a high level of abstraction, rather than to deliver a finished empirical model or domain-specific simulation. The preceding sections develop, in order, the formal ontology and operator toolkit, derived cognitive phenomena, a task-instantiation schema, and compatibility mappings to major frameworks. The mappings section is the culminating technical step: it shows how existing theories can be read as constrained regions of one graph-theoretic language.

\subsection{Implications and scope}

NEST separates a stable representational core from task- and domain-specific instantiations. Future application work should select subgraphs, relation subsets, thresholds, and control policies appropriate to a domain rather than extend the foundational ontology. This division matches longstanding practice in cognitive architecture research \cite{Lieto2018,Newell1990}, where fixed mechanisms combine with specialized knowledge structures and control policies.

\subsection{Limitations}

The present formulation has several limitations. We have not specified a learning algorithm; belief update is defined abstractly. We assume a clear separation between durable belief and capacity-limited working memory, which may not map directly onto neural substrates. Compatibility mappings are variable-level sketches rather than executable comparative simulations. These limitations are deliberate: the contribution is a transparent substrate for later algorithmic, empirical, and domain-specific development.

\subsection{Open research directions}

Several directions follow naturally from the formal core established here.

\textbf{Domain instantiations.} Concrete instantiations in physics reasoning \cite{Chi1992,diSessa2014}, diagnostic inference under uncertainty, and evaluative argumentation can turn the task-instantiation schema into fully specified models with simulations and testable predictions.

\textbf{Computational realization.} An implementation would require concrete data structures for \(B\) and \(W_t\), graph-matching procedures for recognition, and chosen parameterizations of awareness and control over trajectories. Conceptual change could then be studied as learnable graph-rewrite behavior rather than as an abstract update interface alone.

\textbf{Empirical anchoring.} Behavioral measures of capacity and chunking \cite{ChaseSimon1973,GobetSimon1998}, connectivity data during task performance, and metacognitive reports \cite{NelsonNarens1990,Koriat2007} can constrain capacity bounds, grounding behavior, and confidence-like readings of structural diagnostics.

\textbf{Comparative and representational extensions.} Mappings in Section~\ref{sec:mappings} can be extended toward structure-preserving simulations of ACT-R, global workspace theory, and related architectures \cite{AndersonEtAl2004ACTR,Baars1988,LairdLebiereRosenbloom2017}. Probabilistic or fuzzy graph semantics may capture graded belief and uncertainty without altering the discrete ontology developed here.

Taken together, these directions suggest that NEST can serve as a shared representational language for comparing theories, instantiating domains, and linking symbolic structure to empirical and computational work in cognitive science and artificial intelligence.

\IfFileExists{main.bbl}{%

}{%
  \bibliographystyle{plain}%
  \bibliography{references}%
}

\end{document}